\documentclass[sn-mathphys-num]{sn-jnl}

\usepackage{graphicx}%
\usepackage{multirow}%
\usepackage{amsmath,amssymb,amsfonts}%
\usepackage{amsthm}%
\usepackage{mathrsfs}%
\usepackage[title]{appendix}%
\usepackage{xcolor}%
\usepackage{textcomp}%
\usepackage{manyfoot}%
\usepackage{booktabs}%
\usepackage{algorithm}%
\usepackage{algorithmicx}%
\usepackage{algpseudocode}%
\usepackage{listings}%
\usepackage{subfigure}

\theoremstyle{thmstyleone}%
\newtheorem{theorem}{Theorem}[section]
\newtheorem{proposition}{Proposition}[section]%

\theoremstyle{thmstyletwo}%
\newtheorem{remark}{Remark}%
\newtheorem{lemma}{Lemma}[section]%
\newtheorem{corollary}{Corollary}[section] %

\theoremstyle{thmstylethree}%
\newtheorem{definition}{Definition}[section]%

 \newcommand{\card}{\operatorname{card}}
 
\begin{document}

\title[On Finiteness of Stationary Configurations of the Planar Five-vortex Problem]{On Finiteness of Stationary Configurations of the Planar Five-vortex Problem}


\author[1]{\fnm{Xiang} \sur{Yu}}\email{xiang.zhiy@foxmail.com, yuxiang\_math@tju.edu.cn}

\author[2]{\fnm{Shuqiang} \sur{Zhu}}\email{zhusq@swufe.edu.cn}

\affil[1]{\orgdiv{Center for Applied Mathematics and KL-AAGDM}, \orgname{Tianjin University}, \orgaddress{ \city{Tianjin}, \postcode{300072},  \country{China}}}
\affil[2]{\orgdiv{School of Mathematics}, \orgname{
		Southwestern	University of Finance and Economics}, \orgaddress{ \city{Chengdu}, \postcode{611130}, \country{China}}}


\abstract{ 	The finiteness problem of stationary configurations for the planar  five-vortex problem  is considered in this paper. 
	The numbers of equilibria and rigidly translating configurations are shown to be at most  6 and 24  respectively. The numbers of relative  equilibria and collapse configurations 
	are shown to be finite, except perhaps if 
	the 5-tuple of vorticities  belongs to a given codimension 2 subvariety of the vorticity space.
	In particular, if the vorticities are of the same sign, the number of stationary  configurations is finite.
}




\maketitle

\section{Introduction}\label{sec:introduction}

\indent\par
The \emph{planar $N$-vortex problem} which originated from Helmholtz's work in 1858 \cite{helmholtz1858integrale}, considers the motion of point vortices in a fluid plane. It was given a Hamiltonian formulation by Kirchhoff as follows:
\[
\Gamma_n \dot{\mathbf{r}}_n = J \frac{\partial H}{\partial \mathbf{r}_n} = J \sum_{1 \leq j \leq N, j \ne n} \Gamma_j \Gamma_n \frac{\mathbf{r}_j - \mathbf{r}_n}{|\mathbf{r}_j - \mathbf{r}_n|^2}, \quad n = 1, \ldots, N.
\]
Here, $J = \begin{bmatrix} 0 & 1 \\ -1 & 0 \end{bmatrix}$, 
$\mathbf{r}_n = (x_n, y_n) \in \mathbb{R}^2$, and $\Gamma_n$ $(n = 1, \ldots, N)$ are the positions and vortex strengths (or vorticities) of the vortices, and 
the Hamiltonian is $H = -\sum_{1 \leq j < k \leq N} \Gamma_j \Gamma_k \ln|\mathbf{r}_j - \mathbf{r}_k|$, where $|\cdot|$ denotes the Euclidean norm in $\mathbb{R}^2$. The $N$-vortex problem is a widely used model for providing finite-dimensional approximations to vorticity evolution in fluid dynamics, especially when the focus is on the trajectories of the vorticity centers rather than the internal structure of the vorticity distribution \cite{Newton2001}.

An interesting set of special solutions of the dynamical system are homographic solutions, where the relative shape of the configuration remains constant during the motion. An excellent review of these solutions can be found in \cite{aref2003vortex, Newton2001}. 
Following O'Neil, we refer to the corresponding configurations as \emph{stationary}. The only stationary configurations are equilibria, rigidly translating configurations (where the vortices move with a common velocity), relative equilibria (where the vortices rotate uniformly), and collapse configurations (where the vortices collide in finite time) \cite{o1987stationary}.

Many results on stationary configurations have been obtained by focusing on systems with a small number of vortices and exploring symmetric cases (cf. \cite{aref1979motion, grobli1877specielle, hampton2009finiteness, hampton2014relative, novikov1975dynamics, novikov1979vortex, o1987stationary, o2006minimalpolynomial, o2007relative, palmore1982vortex, roberts1999continuum, synge1949motion, Tsai2020bifurcation, yu2021Finiteness} and the references therein). Notably, we emphasize works whose underlying concepts have proven fruitful in understanding central configurations in celestial mechanics. The equations governing stationary configurations are similar to those describing central configurations in celestial mechanics. By applying the topological methods first introduced for central configurations in celestial mechanics \cite{palmore1975classII}, Palmore gave without proof a lower bound on the number of relative equilibria for $N$ vorticities of the same sign \cite{palmore1982vortex}. O'Neil categorized stationary configurations into the four mentioned classes and initiated the systematic examination of their numbers. Roberts  constructed  a continuum of five-body central configurations in  celestial mechanics,  where a negative mass is included,   and then this construction  naturally led to an extension  for the  
five-vortex relative equilibria \cite{roberts1999continuum}. 
With the algebraic method introduced in \cite{hampton2006finiteness}, Hampton and Moeckel  
showed that the number of four-vortex relative equilibria is at most 74,  provided that no subcollection of the four vortices has vanishing total vorticities \cite{hampton2009finiteness}. Roberts  applied Morse theoretical ideas  to the study of relative equilibria in the planar $n$-vortex problem \cite{roberts2018morse}.

Albouy and Kaloshin introduced a novel method to study the finiteness of five-body central configurations in celestial mechanics \cite{Albouy2012Finiteness}. The first author successfully extended this approach to fluid mechanics. Using this new method, the first author established not only the finiteness of four-vortex relative equilibria for any four nonzero vorticities but also the finiteness of four-vortex collapse configurations for a fixed angular velocity. This represents the first result on the finiteness of collapse configurations for $N \geq 4$ \cite{yu2021Finiteness}.

In this paper, we focus on the finiteness of five-vortex stationary configurations. For equilibria and rigidly translating configurations, O'Neil showed that for generic vorticities, the upper bounds are $(N-2)!$ and $(N-1)!$, respectively \cite{o1987stationary, o2006minimalpolynomial}. We confirm these upper bounds for the all five-vortex system. We apply the singular sequence method developed by the first author in \cite{yu2021Finiteness} to investigate the finiteness of relative equilibria and collapse configurations.


Because of the  continuum of five-vortex  relative equilibria constructed by Roberts \cite{roberts1999continuum}
(see also Section \ref{sec:pri})  and the continuum of five-vortex  collapse
configurations constructed by Novikov and Sedov \cite{novikov1979vortex}, 
the finiteness of relative equilibria and collapse
configurations   can only be expected   for generic vorticities,  and,  in case of collapse
configurations,  fixed $\Lambda$.

 Two configurations are called equivalent if they are related by rotations, translations and dilations in the plane. In the following, all the counts are made for the equivalence classes.
We have proven the following results:

\begin{theorem} \label{thm:Main}
	For the planar  five-vortex  problem with nonzero vorticities  $\Gamma_n$ $(n\in\{1,2,3,4, 5\})$, 
	\begin{enumerate}
		\item  There are  at most 6 equilibria;
		\item   There are at most  24 rigidly translating configurations;
		\item  There are    finitely many relative equilibria  provided that none of 14  polynomial systems  of \eqref{equ:constraint_relative_equilibria} holds;
		\item  For  any given $\Lambda \in \mathbb{ C}^*$, there are  finitely many collapse configurations  provided that none of 7   polynomial systems  of \eqref{equ:constraint_collapse} holds. 
	\end{enumerate}  
\end{theorem}

Please see Definition   \ref{def-2} for the meaning of $\Lambda$. We refer  readers to Section \ref{sec:proof} for Systems \eqref{equ:constraint_relative_equilibria} and \eqref{equ:constraint_collapse},  and to  Definition \ref{def:LI} for the meaning of \(L_{j_1, \ldots, j_n}\).

Denote by $\mathbb{R}^*$ the set of nonzero real numbers.  By Theorem \ref{thm:Main}, those  vorticities that may admit infinitely many  relative equilibria  form a subvariety of  
$(\mathbb{R}^*)^5$, and those  vorticities that may admit infinitely many   collapse configurations  form a subvariety of  $\{  (\Gamma_1, \ldots, \Gamma_5): \Gamma_i \in \mathbb{R}^*,  i=1, \ldots, 5,  \sum_{1\le i<j\le 5} \Gamma_i \Gamma_j =0 \}$.  We further show that the codimension of the subvariety 
is  at least 2.

\begin{theorem}\label{thm:cod_2}
	For any choice of five vorticities $(\Gamma_1, \ldots, \Gamma_5) \in (\mathbb{R}^*)^5\setminus \mathcal{A}$,  where $\mathcal{A}$ is a closed algebraic subset of codimension 2, there are finitely  many  relative equilibria of the  five-vortex problem. 
	
	For any choice of five vorticities  in  
	$\{  (\Gamma_1, \ldots, \Gamma_5): \Gamma_i \in \mathbb{R}^*,  i=1, \ldots, 5,  \sum_{1\le i<j\le 5} \Gamma_i \Gamma_j =0 \} \setminus \mathcal{B},$
	where  $\mathcal{B}$ is a closed algebraic subset of codimension 2, there are finitely  many  collapse configurations   of the  five-vortex problem for any given $\Lambda \in \mathbb{ C}^*$. 
\end{theorem}

The following  result on finiteness  is also proved.  

\begin{theorem}\label{thm:restri}
	Given  five vorticities,   if  $\sum_{j\in J} \Gamma_j\ne 0$ and  $\sum_{j, k\in J, j\ne k} \Gamma_j\Gamma_k\ne 0$ for any nonempty subset  $J$  of  $\{1,2,3,4,5\}$,  there are finitely many  complex central configurations of the  five-vortex problem for any fixed  $\Lambda \in \mathbb{ C}^*$. 
\end{theorem}

	Please see Definition \ref{def:cc} and \ref{def:positivenormalizedcentralconfiguration}  for the meaning of complex  central configurations.

If all vorticities are of the same sign, relative equilibria are the only possible stationary configurations.   Theorem \ref{thm:restri} leads to the following result. 

\begin{theorem}
	Consider the planar  five-vortex  problem with all positive or all negative vorticities  $\Gamma_n$ $(n\in\{1,2,3,4, 5\})$. 
	There are finitely many  stationary configurations.  
\end{theorem}

The paper is structured as follows. In Sect. \ref{sec:basicnotations}, we introduce notations and definitions. In Sect. \ref{sec:E_RT}, we prove results on equilibria and rigidly translating configurations. In Sect. \ref{sec:pri}, we briefly review the singular sequence method and the two-colored diagrams. In Sect. \ref{sec:moreproperty}, we identify constraints when some particular sub-diagrams appear. 
In Sect. \ref{sec:matrixrules}, we construct the  problematic diagrams for the 5-vortex problem. 
We derive constraints on vorticities corresponding to each of the 22 diagrams in Sect. \ref{sec:diagram&constraints} and prove the main results in Sect. \ref{sec:proof}.

\section{Basic notations}\label{sec:basicnotations}
\indent\par
We recall some basic notations on stationary configurations and direct readers to a more comprehensive introduction provided by O'Neil \cite{o1987stationary} and Yu \cite{yu2021Finiteness}.

We represent vortex positions $\mathbf{r}_n \in \mathbb{R}^2$ as complex numbers $z_n \in \mathbb{C}$. The equations of motion are  $\dot{z}_n = \textbf{i}V_n$, where
\begin{equation}\label{vectorfield}
V_n= \sum_{1 \leq j \leq N, j \neq n} \frac{\Gamma_j z_{jn}}{r_{jn}^2}= \sum_{ j \neq n} \frac{\Gamma_j }{{\overline{z}_{jn}}}.
\end{equation}
Here, $z_{jn} = z_n - z_j$, $r_{jn} = |z_{jn}| = \sqrt{z_{jn}{\overline{z}_{jn}}}$, $\textbf{i} = \sqrt{-1}$, and the overbar denotes complex conjugation.

Let $\mathbb{C}^N= \{ z = (z_1,  \ldots, z_N):z_j \in \mathbb{C}, j = 1,  \ldots, N \}$ denote the space of configurations for $N$ point vortex. The collision set is defined as  $\Delta=\{ z \in \mathbb{C}^N:z_j=z_k ~~\emph{for some}~~ j\neq k  \} $. The space of collision-free configurations is given by $\mathbb{C}^N \backslash \Delta$.

\begin{definition}\label{def:LI}
	The following quantities and notations are defined:
	\begin{center}
		$\begin{array}{cc}
		\text{Total vorticity} & \Gamma =\sum_{j=1}^{N}\Gamma_j  \\
		\text{Total vortex angular momentum} & L =\sum_{1\leq j<k\leq N}\Gamma_j\Gamma_k  \\
		\text{ Moment of vorticity }& M =\sum_{j=1}^{N}\Gamma_j z_j \\
		\text{ Angular impulse}& I =\sum_{j=1}^{N}\Gamma_j |z_j|^2=\sum_{j=1}^{N}\Gamma_j z_j{\overline{z}_j} \\
		\text{ Size}& S =\sum_{1\leq j<k\leq N}\Gamma_j\Gamma_k r_{jk}^2. 
		\end{array}$
	\end{center}
	For $J=\{j_1, ..., j_n\}\subset \{1, ..., N\}$, we also define 
	\[ \Gamma_J=\Gamma_{j_1, ..., j_n}=\sum_{j\in J} \Gamma_j, \ L_J=L_{j_1, ..., j_n}=\sum_{j<k, j,k \in J} \Gamma_j \Gamma_k.\] 
\end{definition}

A motion is called homographic if the relative shape remains constant. 
	Following  O'Neil \cite{o1987stationary}, we term  a corresponding  configuration as a  \emph{stationary configuration}.  Equivalently, 
\begin{definition} \label{def-1}
	A configuration $z \in \mathbb{C}^N \backslash \Delta$ is stationary if there exists a constant $\Lambda\in {\mathbb{C}}$ such that
	\begin{equation}\label{stationaryconfiguration}
	V_j-V_k=\Lambda(z_j-z_k), ~~~~~~~~~~ 1\leq j, k\leq N.
	\end{equation}
\end{definition}

 There are  only four kinds of homographic motions,  equilibria, translating with a common velocity, uniformly rotating, and   homographic motions that collapse in finite time.   
	Following \cite{o1987stationary, hampton2009finiteness, yu2021Finiteness},  we term 
	the  stationary
	configurations  corresponding to these four classes of  homographic motions as
	equilibria, rigidly translating configurations, relative equilibria  and collapse configurations. 
	Equivalently,

\begin{definition}\label{def-2}
	\begin{itemize}
		\item[i.] $z \in \mathbb{C}^N \backslash \Delta$ is an \emph{equilibrium} if $V_1=\cdots=V_N=0$.
		\item[ii.] $z \in \mathbb{C}^N \backslash \Delta$ is \emph{rigidly translating} if $V_1=\cdots=V_N=c$ for some $c\in \mathbb{C}\backslash\{0\}$.
		\item[iii.] $z \in \mathbb{C}^N \backslash \Delta$ is  a \emph{relative equilibrium} if there exist constants $\lambda\in \mathbb{R}\backslash\{0\},z_0\in \mathbb{C}$ such that $V_n=\lambda(z_n-z_0),~~~~~~~~~~ 1\leq n\leq N$.
		\item[iv.] $z \in \mathbb{C}^N \backslash \Delta$ is a \emph{collapse configuration} if there exist constants $\Lambda,z_0\in \mathbb{C}$ with $\emph{Im}(\Lambda)\neq0$ such that $V_n=\Lambda(z_n-z_0),~~~~~~~~~~ 1\leq n\leq N$.
	\end{itemize}
\end{definition}

\begin{proposition} \cite{o1987stationary} \label{necessaryconditions}
	Every equilibrium has vorticities satisfying $L  =0$; every rigidly
	translating configuration has vorticities satisfying $\Gamma  =0$
\end{proposition}

\begin{definition}
	A configuration $z$ is equivalent to $z'$ if there exist $a, b \in \mathbb{C}$ with $b \neq 0$ such that $z'_n = b(z_n + a)$ for $1 \leq n \leq N$.

		A configuration is called translation-normalized if its translation freedom is removed, rotation-normalized if its rotation freedom is removed, and dilation-normalized if its dilation freedom is removed. 
		A configuration normalized in translation, rotation, and dilation is termed a \emph{normalized configuration}. 
\end{definition}

	We count the stationary configurations according to the equivalence classes. Counting equivalence classes is the same as counting normalized configurations. Note that the removal of any of these three freedoms can be performed in various  ways.

\section{Equilibria and rigidly translating configurations}\label{sec:E_RT}

\indent\par In this section, we investigate  the number of equilibria and rigidly translating configurations of five 
vorticities, via the minimal polynomial system introduced by O'Neil \cite{o2006minimalpolynomial}.  In particular, we will prove the corresponding results in Theorem \ref{thm:restri}.

Recall that the equilibria and the rigidly translating configurations are solutions of the system
\begin{equation}\label{equ:E&RT}
\sum _{k\ne  j} \frac{\Gamma_k}{z_j-z_k} =c, 
\end{equation}
where $c=0$ corresponds to equilibria and $c\ne 0$  corresponds to the rigidly translating configurations. If the configuration  $(z_1, . . . , z _N)$  is an equilibrium or rigidly translating configuration, 
so is each member of its equivalent class  $(b (z_1 +a), ..., b (z_N +a))$ with $a, b  \in \mathbb{C} $ and $a\ne 0$. Given a linear form $A(z_1, ..., z_N)= a_1 z_1+ ...+ a_Nz_N,$ with   $a_1+...+a_N\ne 0$, there is one and only one value of $a$ for which the members of   the class   satisfy  $A=0$. 
Hence, 
it is convenient to identify the  translation-normalized  equilibria and the rigidly translating configurations with points   of 
$\mathbb{C}^{N}$, satisfying 
\begin{equation}\label{equ:E&RT-1}
A(z_1, ..., z_N)=0, \ \sum _{k\ne  j} \frac{\Gamma_k}{z_j-z_k} =c, \ j=1, ..., N,  
\end{equation}

Let $\zeta$ be a complex variable. O'Neil \cite{o2006minimalpolynomial} observed 
that 
system \eqref{equ:E&RT-1} is equivalent to 
\begin{equation}\label{equ:E&RT-2}
A(z_1, ..., z_N)=0,  \  \sum_{j<k} \frac{\Gamma_j \Gamma_k}{(\zeta-z_j)(\zeta-z_k)} = c \sum_{j} \frac{\Gamma_j }{\zeta-z_j},  \  {\rm for}\ \forall  \zeta \ne z_1, ..., z_N.   
\end{equation}

\subsection{Equilibria} The second equation of system \eqref{equ:E&RT-2} becomes $\sum_{j<k} \frac{\Gamma_j \Gamma_k}{(\zeta-z_j)(\zeta-z_k)} =0$. 
Multiplying it with $(\zeta-z_1)...(\zeta-z_N)$ leads to 
\[  \zeta^{N-2} L + \zeta^{N-3}  f_{1} +\zeta^{N-4}  f_{2} ...+ \zeta  f_{N-3} + f_{N-2}=0,  \ \zeta \in \mathbb{ C},   \]
where $f_{k}=f_{k}(z_1, ..., z_N)$  is a  homogeneous polynomial of degree $k$, $1\le k\le N-2$.  Then $L=0$, and  system \eqref{equ:E&RT-2} with $c=0$ is equivalent to 
\begin{equation}\label{equ:E&RT-3}
A(z_1, ..., z_N)=0,  \  f_1=f_2=...=f_{N-2}=0, \  (z_1, ..., z_N)\notin \Delta. 
\end{equation}
The $(N-1)$ homogeneous polynomials define a projective variety $V_0$ in $\mathbb{P}^{N-1}$, and each point of this variety represents an equivalence class of equilibria. If the variety $V_0$ is zero-dimensional, B\'{e}zout's theorem implies that the number of points in $V_0$ is $(N-2)!$, counted by multiplicity.

\begin{proposition}\cite{o2006minimalpolynomial} \label{prp:ON-e}
	Let the nonzero vorticities $\Gamma_1, \ldots, \Gamma_N$ satisfy the relation $L = 0$, and let $V_0$ be defined as above, so that $V_0 \setminus \Delta$ is the set of all equilibria. Suppose there are two indices $p, q$ such that for all proper subsets $J$ of $\{1, \ldots, N\}$, $\{p, q\} \subset J$ implies $L_J \ne 0$. Then $V_0$ contains exactly $(N-2)!$ points counted according to multiplicity, and there are no more than $(N-2)!$ equilibria.
\end{proposition}

	Please see Definition \ref{def:LI}  for the meaning of $\Gamma_J, \Gamma_{j_1, ..., j_n}$, $L_J$ and $L_{j_1, ..., j_n}$.

\begin{remark}\label{rmk:L}
	There are some facts on $L_J$. If $L_J=0$, then the cardinality of $J$, denoted by  $\card (J)$, can not be two, since $\Gamma_i\ne0$ for all $i=1, ... , N$; If $L_J=0$, then $\Gamma_J\ne 0$, since $\Gamma_J^2>2L_J$; 
	If $L_J=0$, then for any subset $K$ of $J$ with $\card (K)+1=\card (J)$, we have $L_K\ne 0$. Suppose that $K\cup \{1\}=J$, $L_J=L_K=0$. Then the identity
	$L_J=L_K+\Gamma_K\Gamma_1$
	implies that $\Gamma_K=0$, which contradicts with $L_K=0$. 
\end{remark}

Consider the case of $N = 4$. Since $L = 0$, Remark \ref{rmk:L} implies that $L_J \ne 0$ for all subsets $J$ of $\{1,2,3,4\}$. Hence, Proposition \ref{prp:ON-e} implies that any four vorticities satisfying $L = 0$ have at most 2 distinct equilibria. In fact, there are always exactly 2 distinct equilibria, as shown by O'Neil \cite{o1987stationary} and Hampton and Moeckel \cite{hampton2009finiteness}. We now utilize Proposition \ref{prp:ON-e} to study the case of $N = 5$.

\begin{proposition} 
	Let the nonzero  vorticities $\Gamma_1, ..., \Gamma_5$ satisfy the relation $L=0$. Then  there are at most 6 equilibria. 
\end{proposition}

\begin{proof}
	We assume that for any pair of indices, there is some proper subset $J$ of $\{1, 2, 3, 4, 5\}$ containing them such that  $L_J=0$,  and arrive at  a contradiction. Since $L=0$, then $L_J=0$ holds only if  the cardinality of $J$  is three, as noted  in Remark \ref{rmk:L}.

	Since $L=0$, the vorticities cannot all have the same sign.  We  divide the discussion  into two cases: one with a single negative vorticity and the other with two negative vorticities.

	\emph{Case I: The signs of the five vorticities are $(+,+,+,+,-)$.} Consider the pair $\{1,2\}$. It must hold that $L_{125} = 0$. Similarly, for the pairs $\{1,3\}$, $\{1,4\}$, $\{2,3\}$, $\{2,4\}$, and $\{3,4\}$, we have:
	\[
	L_{125} = L_{135} = L_{145} = L_{235} = L_{245} = L_{345} = 0.
	\]
	It is straightforward to find the solution: $\Gamma_1 = \Gamma_2 = \Gamma_3 = \Gamma_4 = -2\Gamma_5$, which contradicts the equation $L = 0$.
	
	\emph{Case II: The signs of the five vorticities are $(+,+,+,-,-)$.} Consider the three pairs $\{1,2\}$, $\{1,3\}$, and $\{2,3\}$. The following system holds:
	\[
	L_{124} L_{125} = 0, \quad L_{134} L_{135} = 0, \quad L_{234} L_{235} = 0.
	\]
	It is sufficient to consider two sub-cases: when there are three or two $5$'s in the above system.

	\emph{Sub-case II-1: $L_{125} = L_{135} = L_{235} = 0$.}
	It is straightforward to find the solution: $\Gamma_1 = \Gamma_2 = \Gamma_3 = -2\Gamma_5$. Now consider the pair $\{4,5\}$. It must hold that $L_{145} = 0$. However,
	\[
	L_{145} = \Gamma_1 \Gamma_5 + \Gamma_4 (\Gamma_1 + \Gamma_5) = -2 \Gamma_5^2 - \Gamma_4 \Gamma_5 < 0,
	\]
	which is a contradiction.
	
	\emph{Sub-case II-2: $L_{125} = L_{135} = L_{234} = 0$.}
	It is straightforward to find that $\Gamma_2 = \Gamma_3 = -2\Gamma_4$. Now consider the pair $\{4,5\}$, then $L_{145} L_{245} L_{345} = 0$. Since
	\[
	L_{245} = \Gamma_2 \Gamma_4 + \Gamma_5 (\Gamma_2 + \Gamma_4) = -2 \Gamma_4^2 - \Gamma_4 \Gamma_5 < 0,
	\]
	it must hold that $L_{145} = 0$. If $L_{145} = 0$, then by $L_{125} = L_{145}$, we obtain $\Gamma_4 = \Gamma_2$, which is a contradiction.

	By Proposition \ref{prp:ON-e}, for any group of five  vorticities  with $L=0$,   the variety $V_0$ is zero-dimensional and there are  at most 6 equilibria. 
\end{proof}

The above result has also been proved differently by Tsai \cite{Tsai2020bifurcation}, where a complete bifurcation diagram is provided. The actual number of equilibria could be fewer or even zero, as some solutions may have multiplicity greater than one, and some points of $V_0$ may lie on $\Delta$. If $V_0 \cap \Delta \neq \emptyset$, then there exists a proper subset $J$ such that $L_J = 0$ \cite{o2006minimalpolynomial}.

For instance, if  $L_{123}=0$, then $\Gamma_{123}\ne 0$,  we can construct  point   on  $V_0\bigcap \Delta$ using \eqref{equ:E&RT-2}.  Let 
\[  J=\{1,2,3\}, \ d_j = \zeta -z_j, \  A(z_1, ..., z_5)=z_1, \  z_1=z_2=z_3. \] 
Then  
\begin{align*}
	0=\sum_{j<k} \frac{\Gamma_j \Gamma_k}{d_j d_k}= & \ \ \sum_{j<k, j, k \in J} \frac{\Gamma_j \Gamma_k}{\zeta^2} +\sum_{k=4}^5 \frac{\Gamma_J \Gamma_k}{\zeta d_k} +\frac{\Gamma_4\Gamma_5}{d_4 d_5}  \\
	&=\frac{\Gamma_4}{d_4} ( \frac{\Gamma_5}{z_4- z_5} + \frac{\Gamma_J}{z_4 }  )  +\frac{\Gamma_5}{d_5} ( \frac{\Gamma_4}{z_5- z_4} + \frac{\Gamma_J}{z_5 }  ) -\frac{\Gamma_J}{\zeta} ( \frac{\Gamma_4}{z_4} +\frac{\Gamma_5}{z_5}).  
\end{align*}
The above equation holds if and only if $z_4=-\frac{\Gamma_4 z_5}{\Gamma_5}$ and $L=0$.  
Hence we find the following point of $V_0$ on $\Delta$
$ (0:0:0:-\frac{\Gamma_4 }{\Gamma_5}:1),$
with multiplicity 2. Thus, there are at most 4 equilibria for five vorticities satisfying $L = L_{123} = 0$. Similar solutions can be constructed for $N$ vorticities satisfying $L = L_J = 0$, where the cardinality of $J$ is $N-2$.

Assume there are four vorticities that satisfy $L_{123} = L_{124} = L_{134} = 0$. We can choose $\Gamma_5$ such that $L = 0$. The example above shows that there are three distinct solutions of \eqref{equ:E&RT-2} on $\Delta$, each with multiplicity 2. Hence, there is no equilibrium for this group of vorticities. Up to renumbering, the vorticities are:
\[
\Gamma_1 = 1, \ \Gamma_2 = -2, \ \Gamma_3 = -2, \ \Gamma_4 = -2, \ \Gamma_5 = \frac{6}{5}.
\]
Examples of groups of vorticities that do not admit any equilibrium exist for any $N \ge 5$ (see \cite{Tsai2020bifurcation}).

\subsection{ Rigidly translating configurations} Multiplying  $(\zeta-z_1)...(\zeta-z_N)$  to the second equation of system \eqref{equ:E&RT-2} leads to $\Gamma=0$ and $L=c \sum_{ j=1}^N \Gamma_j z_j$. Then we have  $(\sum_{ j=1}^N \Gamma_j z_j)\sum_{j<k} \frac{\Gamma_j \Gamma_k}{(\zeta-z_j)(\zeta-z_k)} =L \sum_{j} \frac{\Gamma_j }{\zeta-z_j}$, which leads to 
\[  \zeta^{N-1} L\Gamma   +\zeta^{N-3}  g_{1} ...+ \zeta  g_{N-3} + g_{N-2}=0,  \ \zeta \in \mathbb{ C},   \]
where $g_{k}=g_{k}(z_1, ..., z_N)$  is a  homogeneous polynomial of degree $k$, $1\le k\le N-2$.  Note that $ \sum_{ j=1}^N \Gamma_j z_j\ne 0$,  as otherwise  $L=c \sum_{ j=1}^N \Gamma_j z_j=0$, which contradicts $\Gamma = 0$. 
Then  system \eqref{equ:E&RT-2} with $c\ne0$ is equivalent to 
\begin{equation}\label{equ:E&RT-4}
A(z_1, ..., z_N)=0,  \  g_1=g_2=...=g_{N-2}=0,  \ (z_1, ..., z_N)\notin \Delta. 
\end{equation}
The $N-1$ homogeneous polynomials  define a projective variety $V_c$ in $\mathbb{P}^{N-1}$ and each point of this variety represents an equivalent class of rigidly translating configurations.  If the variety $V_c$ is zero dimensional,  B\'{e}zout theorem implies that the number of points in $V_c$ is $(N-1)!$,  counted by multiplicity.

\begin{proposition}
	\cite{o2006minimalpolynomial} \label{prp:ON-r}
	Let the nonzero vorticities  $\Gamma_1, ..., \Gamma_N$ satisfy the relation $\Gamma=0$, and let $V_c$ be defined as above so that $V_c\setminus \Delta $ is the set of all rigidly translating configurations.  Suppose that  there are two indices $p, q$ such that for all proper subsets $J$ of $\{1, . . . , N\},  \{p, q\}\subset J$ implies $L_J\ne 0$ and $\Gamma_J\ne 0$. Then $V_c$ contains exactly $(N - 1)!$ points, counted according to multiplicity, and there are no more than $(N -1)!$ rigidly translating configurations.
\end{proposition}

For $N=4$, if the nonzero  vorticities $\Gamma_1, ..., \Gamma_4$ satisfy the relation $\Gamma=0$, then there are at most 6 rigidly translating configurations. This result  has already been proved  Hampton and Moeckel with
resultant theory  in \cite{hampton2009finiteness}. 
It is easy to  apply the above criterion to obtain  an alternative proof. However, it would be similar to the one of Proposition \ref{prop:rigid_5} and simpler, so we omit it.

Note that for the  vorticities $(1, -1, a, -a)$, there is always one solution in $\Delta$ with multiplicity  three. Then there are at most three rigidly translating configurations.  In particular, for $a=1$, there are two  solutions in $\Delta$   with multiplicity three, so there are no rigidly translating configurations \cite{o2006minimalpolynomial, Tsai2020bifurcation}.

\begin{proposition} \label{prop:rigid_5}
	Let the nonzero  vorticities $\Gamma_1, ..., \Gamma_5$ satisfy the relation $\Gamma=0$. Then the  dimension of $V_c$ is  zero, so there are at most 24 rigidly translating configurations. 
\end{proposition}

\begin{proof}	
	We assume that for any pair of indices, there exist some proper subsets $J_1, J_2$ of $\{1, 2, 3, 4, 5\}$  containing them such that  $L_{J_1}\Gamma_{J_2}=0$,   and  we derive   contradictions. Note that $\Gamma_J=0$ holds only if  the cardinality of $J$  is two or  three, and that   $L_J=0$ holds only if  the cardinality of $J$ is three or four,  as noted in Remark \ref{rmk:L}.

	Since $\Gamma=0$, the vorticities can not be of the same sign. We can simplify the discussion by dividing it into two cases: one with a single negative vorticity and the other with two negative vorticities.

	\emph{ Case I: the signs of the five vorticities be  $(+,+,+,+, -)$. } 
	Then $-\Gamma_5=\Gamma_{1234}$ and   $\Gamma_J\ne0$ for any proper subset $J$.  Consider the pair $\{1,2\}$. Without loss of generality, we obtain $L_{125}L_{1235}=0$. 
	
	If $L_{125}=L_{12}+\Gamma_5\Gamma_{12}=0$, then  
	$-\Gamma_5= \frac{L_{12}}{\Gamma_{12}} < \Gamma_{12}, $
	which is a contradiction. 
	
	If $L_{1235}=L_{123}+\Gamma_5\Gamma_{123}=0$, then  
	$-\Gamma_5= \frac{L_{123}}{\Gamma_{123}} < \Gamma_{123}, $
	which is a contradiction.

	\emph{Case II: the signs of the five vorticities be   $(+,+,+,-, -)$. }  
	There are two sub-cases: there exists some  proper subset  $J$ such that $\Gamma_J=0$ or not.

	\emph{ Sub-case II-1: there are proper  subsets $J$ such that $\Gamma_J= 0$.  }  Then we may assume that  the vorticities are 
	$(1, b, c, -1,-(b+c))$ where $b, c>0$. Then  $\Gamma_{14}=\Gamma_{235}=0$. Note that $\Gamma_{124}=b, \Gamma_{134}=c$, and 
	\begin{align*}
		&L_{125}=-c - b(b+c)<0,  &L_{134}=-1, ~~~ \  \ \ &L_{135}=-b - c(b+c)<0, \\
		&L_{1235}=-b^2 - b c - c^2<0,  &L_{145}=-1, ~~~ \  \ \ &L_{1245}=-1 - b(b+c)<0, \\
		&L_{1345}=-1 - c(b+c)<0,  &L_{124}=-1, ~~~ \  \ \   &L_{2345}=-b^2 - b c - c^2<0. 
	\end{align*}
	
	Consider the pair of indices $\{4,5\}$. Since $\Gamma_J\ne 0$ for any proper subset containing the two indices, then the following system holds:
	\[ L_{245}L_{345}=0,\]
	since it is already known that $L_{145}, L_{1245}, L_{1345}, L_{2345}$ are all negative. 
	Similarly, consider the four pairs  $\{1,2\}, \{1,3\}, \{1,5\}$.  Then   the following system holds
	\[ \Gamma_{125} L_{1234}=0,\  \Gamma_{135} L_{1234}=0,\ \Gamma_{125}\Gamma_{135}=0.\]
	It is enough to divide the discussion into two sub-cases: $L_{1234}=0$ or not.

	\emph{Sub-case II-1-1: $L_{1234}=bc-1=0$. }  Then by $\Gamma_{125}\Gamma_{135}=(1-b)(1-c)=0$, we have $b=1$ or $c=1$. By $L_{1234}=0$, it follows that $b=1$ and $c=1$. However, this contradicts with $ L_{245}L_{345}=[c-b(b+c)][b-c(b+c)]=0$.

	\emph{Sub-case II-1-2:  $L_{1234}\ne 0$. }  Then $\Gamma_{125}=0$  and $\Gamma_{135}=0$. It follows that $b=1$ and $c=1$. However, this contradicts with $ L_{245}L_{345}=0$.

	\emph{ Sub-case II-2: $\Gamma_J\ne 0$ for any proper subset $J$.} 
	We assume that $\Gamma_{123}=1$, $\Gamma_1\le \Gamma_2\le \Gamma_3$,  $\Gamma_{45}=-1$ and  $|\Gamma_4|\le |\Gamma_5|$.  Then 
	\[ L_{123}\le \frac{1}{3}, \ L_{23}\le\frac{1}{4}\Gamma^2_{23}<\frac{1}{4}\Gamma_{23},  \  L_{45}\le \frac{1}{4}, \Gamma_{3}\ge \frac{1}{3},  \Gamma_{5}\le -\frac{1}{2}.\]
	So, 
	\begin{align*}
		&L_{1235}=L_{123}+\Gamma_{5}\Gamma_{123}= L_{123}-\frac{1}{2}<0,\\
		&L_{345}=L_{45}+\Gamma_{3}\Gamma_{45}= L_{45}-\Gamma_{3}<0,\\
		&L_{235}=L_{23}+\Gamma_{5}\Gamma_{23}\le L_{23}-\frac{1}{2}\Gamma_{23}<0. 
	\end{align*}
	Similarly, it holds that $L_{125}<0, L_{135}<0.$

	Consider the pair $\{3,5\}$. All possible $L_{J}=0$ with $J$ containing the pair  are $L_{1345}, L_{2345}$. Then it must hold $L_{1345} L_{2345}=0$. 
	Consider also the two other pairs $\{1,5\}$ and $\{2,5\}$.  Then  the following system must hold, 
	\begin{equation} \label{equ:RT5}
	L_{1345} L_{2345} =0, \ L_{145}  L_{1245} L_{1345} =0, \ L_{245}  L_{1245} L_{2345} =0.  \  
	\end{equation}
	There are five sub-cases. 
	
	\emph{Sub-case II-2-1:}  $L_{1345}=L_{1245}=0$.  Then $\Gamma_2=\Gamma_3$ and $\Gamma_{13}>\frac{1}{2}$. Thus, 
	\[ L_{1345}= -\Gamma_{13} +L_{13}+L_{45} <-\frac{1}{2} +\frac{1}{4}+\frac{1}{4}<0, \]
	which is a contradiction.

	\emph{Sub-case II-2-2:}  $L_{1345}=L_{245}=0$.  The second equation implies that $\Gamma_2=L_{45}\le \frac{1}{4}$. Then $\Gamma_{13}\ge \frac{3}{4}$, and 
	\[  L_{1345}= L_{13}+L_{45}- \Gamma_{13}\le \frac{1}{4}+\frac{1}{4}- \frac{3}{4}<0,   \]
	which is  a contradiction.

	\emph{Sub-case II-2-3:}   $L_{1345}=L_{2345}$.   Then $\Gamma_1=\Gamma_2$ and $\Gamma_{13}\ge \frac{2}{3}$. Thus, 
	\[ L_{1345}= -\Gamma_{13} +L_{13}+L_{45} \le -\frac{2}{3} +\frac{1}{4}+\frac{1}{4}<0, \]
	which is  a contradiction.

	\emph{Sub-case II-2-4:}  $L_{2345}=L_{145}=0$.   The second equation implies that $\Gamma_1=L_{45}\le \frac{1}{4}$. Then $\Gamma_{23}\ge \frac{3}{4}$, and 
	\[  L_{2345}= L_{23}+L_{45}- \Gamma_{23}\le \frac{1}{4}+\frac{1}{4}- \frac{3}{4}<0,   \]
	which is  a contradiction.

	\emph{Sub-case II-2-5:}  
	$L_{2345}=L_{1245}=0$.  Then $\Gamma_1=\Gamma_3$ and $\Gamma_{1}=\Gamma_{2}=\Gamma_{3}= \frac{1}{3}$. Thus, 
	\[ L_{2345}= -\Gamma_{23} +L_{23}+L_{45} \le  -\frac{2}{3} +\frac{1}{9}+\frac{1}{4}<0, \]
	which is  a contradiction.

	By Proposition \ref{prp:ON-r}, for any group of five vorticities  with $\Gamma=0$,   the variety $V_c$ is zero-dimensional and there are  at most 24 rigidly translating configurations. 
\end{proof}

\section{Singular sequences for central configurations and coloring rules}\label{sec:pri}

\indent\par
In this section, we briefly review   the basic elements of the Albouy-Kaloshin approach developed   by Yu  \cite{yu2021Finiteness} for the finiteness of relative equilibria and collapse configurations, including,  among others, the notation of central configurations, the extended system, the notation of singular sequences, the two-colored diagrams, and the rules for  the two-colored diagrams. For  a more comprehensive introduction, please refer to   \cite{yu2021Finiteness}.

\subsection{Central configurations of the planar N-vortex problem}
Recall Definition \ref{def-2}. Equations of relative equilibria and collapse configurations share the form:
\begin{equation}\label{stationaryconfiguration1}
V_n=\Lambda(z_n-z_0),~~~~~~~~~~ 1\leq n\leq N,
\end{equation}
where $\Lambda\in \mathbb{R}\backslash\{0\}$  indicates relative equilibria and $\Lambda\in \mathbb{C}\backslash\mathbb{R}$  indicates  collapse configurations.
\begin{definition}\label{def:cc}
	Relative equilibria and collapse configurations are both called \emph{central configurations}.
\end{definition}

The equations (\ref{stationaryconfiguration1}) read 
\begin{equation}\label{stationaryconfiguration2}
\Lambda z_n= V_n,~~~~~~~~~~ 1\leq n\leq N,
\end{equation}
if the translation freedom is removed, i.e., we substitute $z_n$ with $z_n + z_0$ in equations (\ref{stationaryconfiguration2}). The solutions then satisfy:
\begin{equation}\label{center0}
M=0, \ \Lambda I= L.
\end{equation}
To remove dilation freedom, we enforce $|\Lambda| = 1$.

Introduce a new set of variables $w_n$ and a ``conjugate"
relation:
\begin{equation}\label{stationaryconfiguration3}
\Lambda z_n=\sum_{ j \neq n} \frac{\Gamma_j }{{w_{jn}}},\ \ 
\overline{\Lambda} w_n=\sum_{ j \neq n} \frac{\Gamma_j }{{z_{jn}}},\ \ \  1\leq n\leq N,
\end{equation}
where $z_{jn}=z_{n}-z_{j}$ and $w_{jn}=w_{n}-w_{j}$.

The rotation symmetry of \eqref{stationaryconfiguration2} leads to   the invariance of
(\ref{stationaryconfiguration3}) under the map 
\[R_a: (z_1, ..., z_n, w_1, ..., w_n) \mapsto (az_1, ..., az_N, a^{-1} w_1, ..., a^{-1}w_N)\]
for any $a\in\mathbb{C}\backslash \{0\}$. 

Introduce  the variables $Z_{jk},W_{jk}\in \mathbb{C}$ $(1\leq j< k\leq N)$ such that
$Z_{jk}=1/w_{jk}, W_{jk}=1/z_{jk}$. For $1\leq k< j\leq N$ we set $Z_{jk}=-Z_{kj}, W_{jk}=-W_{kj}$. Then equations (\ref{stationaryconfiguration2}) together with the condition $z_{12}\in\mathbb{R}$ and $|\Lambda|=1$ are embedded into the following extended system
\begin{equation}\label{equ:complexcc}
\begin{array}{cc}
\Lambda z_n=\sum_{ j \neq n} \Gamma_j Z_{jn},&1\leq n\leq N, \\
\overline{\Lambda} w_n=\Lambda^{-1} w_n=\sum_{ j \neq n} \Gamma_j W_{jn},& 1\leq n\leq N, \\
Z_{jk} w_{jk}=1,&1\leq j< k\leq N, \\
W_{jk} z_{jk}=1,&1\leq j< k\leq N, \\
z_{jk}=z_k-z_j,~~~  w_{jk}=w_k-w_j,&1\leq j, k\leq N, \\
Z_{jk}=-Z_{kj},~~~ W_{jk}=-W_{kj},&1\leq k< j\leq N, \\
z_{12}=w_{12}.
\end{array}
\end{equation}
This is a polynomial system in the  variables $\mathcal{Q}=(\mathcal{Z},\mathcal{W})\in\mathbb{C}^{2\mathfrak{N}}$, here
\begin{center}
	$\mathcal{Z}=(\mathcal{Z}_{1},\mathcal{Z}_{2},\ldots,\mathcal{Z}_{\mathfrak{N}})=(z_1,z_2,\ldots,z_N,Z_{12},Z_{13},\ldots,Z_{(N-1)N})$,
	$\mathcal{W}=(\mathcal{W}_{1},\mathcal{W}_{2},\ldots,\mathcal{W}_{\mathfrak{N}})=(w_1,w_2,\ldots,w_N,W_{12},W_{13},\ldots,W_{(N-1)N})$.
\end{center}
and $\mathfrak{N}=N(N+1)/2$.  

\begin{definition}\label{def:positivenormalizedcentralconfiguration}
 A \emph{complex normalized} central configuration of the planar
		$N$-vortex problem is a solution of (\ref{equ:complexcc}).  A \emph{real  normalized} central configuration of the planar
		$N$-vortex problem is a  complex normalized central configuration  satisfying $z_n={\overline{w}}_n$ for any $n=1, \ldots, N$.  
\end{definition}

 Note that a real  normalized central configuration of Definition \ref{def:positivenormalizedcentralconfiguration} is exactly  a central configuration of 
	Definition \ref{def:cc}.  	
We will use the name ``distance" for the $r_{jk}=\sqrt{z_{jk}{w_{jk}}}$. Strictly speaking, the distances $r_{jk}=\sqrt{z_{jk}{w_{jk}}}$ are now bi-valued. However, only the squared distances appear in the system, so we shall  understand $r^2_{jk}$ as $z_{jk}w_{jk}$ from now on.

\subsection{Singular sequences}
Let $\|\mathcal{Z}\|=\max_{j=1,2,\ldots,\mathfrak{N}}|\mathcal{Z}_{j}|$ be the modulus of the maximal component of
the vector $\mathcal{Z}\in \mathbb{C}^\mathfrak{N}$. Similarly, set  $\|\mathcal{W}\|=\max_{k=1,2,\ldots,\mathfrak{N}}|\mathcal{W}_{k}|$.

One important feature of System \eqref{equ:complexcc} is the symmetry: 
if $\mathcal Z, \mathcal W$ is a solution, so is $ a\mathcal Z, a^{-1} \mathcal W$ for any $a\in\mathbb{C}\backslash\{0\}$. Thus, we can replace 
the  normalization $z_{12}=w_{12}$  in System \eqref{equ:complexcc}  by 
$\|\mathcal{Z}\|=\|\mathcal{W}\|$.  From now on, we consider System \eqref{equ:complexcc} with this new normalization.

Consider
a sequence $\mathcal{Q}^{(n)}$, $n=1,2,\ldots$, of solutions  of (\ref{equ:complexcc}).  Take a sub-sequence such that the maximal component of $\mathcal{Z}^{(n)}$ is fixed, i.e., there is a $j\in \{1,2,\ldots,\mathcal{N}\}$ that is  independent 
of $n$ such that  $\|\mathcal{Z}^{(n)}\|=|\mathcal{Z}^{(n)}_{j}|$. 
Extract again in such a way that the sequence $\mathcal{Z}^{(n)}/\|\mathcal{Z}^{(n)}\|$ converges.
Extract again  in such a way that  the maximal component of $\mathcal{W}^{(n)}$ is fixed. Finally,  extract  in
such a way that the sequence $\mathcal{W}^{(n)}/\|\mathcal{W}^{(n)}\|$ converges.

\begin{definition}[Singular sequence]
	Consider a sequence of complex normalized central configurations with the property  that $\mathcal{Z}^{(n)}$ is unbounded. A
	sub-sequence extracted by the above process is called
	a \emph{singular sequence}.
\end{definition}

\begin{lemma}\label{Eliminationtheory}\cite{Albouy2012Finiteness}
	Let $\mathcal{X}$ be a closed algebraic subset of $\mathbb{C}^m$ and $f:\mathbb{C}^m\rightarrow \mathbb{C}$ be a
	polynomial. Either the image $F(\mathcal{X})\subset\mathbb{ C}$ is a finite set, or it is the complement
	of a finite set. In the second case one says that f is dominating.
\end{lemma}

\subsection{The  two-colored diagrams} \label{sec:rule}

For two sequences of non-zero numbers, $a, b$, we use $a\sim b$,  $a\prec b$, $ a\preceq b$,  and $ a \approx b$ to represent ``$a/b\rightarrow 1$'', ``$a/b\rightarrow 0$'', ``$a/b$ is bounded'' and ``$a\preceq b$, $a\succeq  b$'' respectively.  

Recall that  a  singular sequence satisfy the property   $\|\mathcal{Z}^{(n)}\|=\|\mathcal{W}^{(n)}\|\to \infty$.  Set $\|\mathcal{Z}^{(n)}\|=\|\mathcal{W}^{(n)}\|=1/\epsilon^2$. Then $\epsilon\rightarrow 0$.  Following Albouy-Kaloshin, \cite{Albouy2012Finiteness}, the \emph{two-colored diagram} was introduced in \cite{yu2021Finiteness} to  classify the singular sequences.  Given a singular sequence, the indices of the vertices  will be written down. 
The first color, called the $z$-color (red),   is used to mark the maximal order components of $\mathcal{Z}$. If $z_k\approx \epsilon^{-2}$,  draw a $z$-circle around the
vertex $\textbf{k}$; If   $Z_{jk}\approx \epsilon^{-2}$,  draw a $z$-stroke between vertices $\textbf{k}$ and $\textbf{j}$. 
They constitute  the $z$-diagram. 
The second color, called the $w$-color (blue and dashed),   is used to mark the maximal order components of $\mathcal W$ in similar manner. Then we also have the  $w$-diagram. The two-colored diagram is the combination of the  $z$-diagram and the $w$-diagram,  see Figure \ref{fig:edges}.

If there
is either a $z$-stroke, or a $w$-stroke, or both between vertex $\textbf{k}$ and vertex $\textbf{l}$, we say that there is an edge between them.  There are three types of edges,
$z$-edges, $w$-edges and $zw$-edges, see Figure \ref{fig:edges}. 

\begin{figure}[h!]
	\centering
	\includegraphics[width=0.7\textwidth]{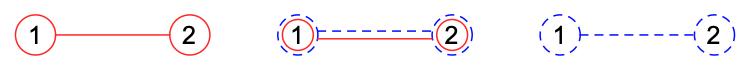} 
	\caption{On the left, vertices \textbf{1, 2} are $z$-circled, and a $z$-edge is between them; In the middle, vertices \textbf{1, 2} are $z$- and $w$-circled, and a $zw$-edge is between them; On the right,  vertices \textbf{1, 2} are $w$-circled, and a $w$-edge is between them. 
	}
	\label{fig:edges}
\end{figure}

The following concepts were introduced to characterize some features of singular sequences.    In the $z$-diagram,   vertices  $\textbf{k}$ and $\textbf{l}$
are  called \emph{$z$-close},  if $z_{kl} \prec\epsilon^{-2}$; 
a $z$-stroke between  vertices   $\textbf{k}$ and  $\textbf{l}$ is  called  	a   \emph{maximal $z$-stroke} if $z_{kl} \approx  \epsilon^{-2}$; 
a subset of vertices are called \emph{an isolated component of the $z$-diagram} if there is no $z$-stroke  between a vertex of this subset and  a vertex of its complement. 
These concepts also apply to the $w$-diagram.

\begin{proposition}[Estimate]\label{Estimate1}\cite{yu2021Finiteness}
	For any $(k,l)$, $1\leq k<l\leq N$, we have $\epsilon^2\preceq z_{kl}\preceq \epsilon^{-2}$, $\epsilon^2\preceq w_{kl}\preceq \epsilon^{-2}$ and $\epsilon^2\preceq r_{kl}\preceq \epsilon^{-2}$.
	
	There is a $z$-stroke between $\textbf{k}$ and $\textbf{l}$ if and only if $w_{kl}\approx \epsilon^{2}$. Then $ r_{kl}\preceq 1$.
	
	There is a maximal $z$-stroke between $\textbf{k}$ and $\textbf{l}$ if and only if $z_{kl}\approx \epsilon^{-2}, w_{kl}\approx \epsilon^{2}$. Then $ r_{kl}\approx1$.
	
	There is a $z$-edge between $\textbf{k}$ and $\textbf{l}$ if and only if $z_{kl}\succ \epsilon^{2},w_{kl}\approx \epsilon^{2}$. Then $\epsilon^{2}\prec r_{kl}\preceq 1 $.
	
	There is a maximal $z$-edge between $\textbf{k}$ and $\textbf{l}$ if and only if $z_{kl}\approx \epsilon^{-2},w_{kl}\approx \epsilon^{2}$. Then $ r_{kl}\approx 1$.
	
	There is a $zw$-edge between $\textbf{k}$ and $\textbf{l}$ if and only if $z_{kl},w_{kl}\approx \epsilon^{2}$. This can be  characterized as $ r_{kl}\approx \epsilon^{2}$.
\end{proposition}

\begin{remark}
	By the estimates above, the strokes in a $zw$-edge are not maximal. A maximal $z$-stroke is exactly a maximal $z$-edge.
\end{remark}

The following rules for the two-colored diagrams  are valid if  ``$z$'' and ``$w$''  were switched.

 \medskip
	\emph{Rule I}: 
	There is something at each end of any $z$-stroke: another $z$-stroke
	or/and a $z$-circle drawn around the name of the vertex. A $z$-circle cannot be isolated; there must be a $z$-stroke emanating from it. There is at least one
	$z$-stroke in the $z$-diagram. \medskip

	\emph{Rule II}: If vertices $\textbf{k}$ and $\textbf{l}$
	are  $z$-close, they are both $z$-circled or both not
	$z$-circled. \medskip

	\emph{Rule III}: The moment of vorticity of a set of vertices forming an isolated component of the $z$-diagram is $z$-close to the origin. \medskip
	
	\emph{Rule IV}: Consider the $z$-diagram or an isolated component of it. If there
	is a $z$-circled vertex, there is another one.   If the $z$-circled vertices are all
	$z$-close together,  the total vorticity of these $z$-circled  vertices is zero. \medskip
	
	\emph{Rule V}:  There is at least one $z$-circle at certain end of any maximal $z$-stroke. As a result,
	if an isolated component of the $z$-diagram has no $z$-circled vertex,
	then it has no maximal $z$-stroke. \medskip
	
	\emph{Rule VI}: 
	If there are two consecutive $z$-stroke, there is a third $z$-stroke closing the triangle. \medskip

	\begin{remark}\label{rmk:rules}
		We would like to compare our rules for the \(N\)-vortex problem with those of Albouy and Kaloshin for the \(N\)-body problem \cite{Albouy2012Finiteness}. Our Rules I to V correspond to their Rules 1a to 1e but are weaker due to the possibility of negative vorticities and clusters with zero total vorticity. Our Rule VI is similar to their Rule 2b but is stronger in the sense that we only need two consecutive \(z\)-strokes (\(w\)-strokes) to form a triangle of \(z\)-strokes (\(w\)-strokes), whereas they require two consecutive \(zw\)-edges to form a triangle of \(zw\)-edges. 
		
		We lack the two-color rules due to the differences between the two potentials. The Newtonian potential leads to \(Z_{ij} = z_{ij}^{-\frac{1}{2}} w_{ij}^{-\frac{3}{2}}\), so  the existence of a \(z\)-stroke provides estimates for both \(z_{ij}\) and \(w_{ij}\), which in turn yields the two-color rules. In contrast, the logarithmic potential leads to \(Z_{ij} = w_{ij}^{-1}\), so the existence of a \(z\)-stroke provides an estimate only for \(w_{ij}\), thus we lack the two-color rules.
	\end{remark}

\emph{Example: Robert's continuum:}  Roberts found the following continuum of relative equilibria for five vortices. The vorticities are $\Gamma_1=-1, \Gamma_2=\Gamma_3 =\Gamma_4=\Gamma_5=2$. The first vortex is at the origin $z_1=w_1=0$, while the other four vortices form a rhombus
\[   (z_2,z_3, z_4, z_5)=  (a, b, -a, -b), \ (w_2,w_3, w_4, w_5)=  (a, -b, -a, b), \  \]
where $a\in \mathbb{R}$ and $\textbf{i} b\in \mathbb{R}$ for the real configurations. One can check that the configuration defined by the coordinates  above and   the restriction $r^2_{23}=r^2_{34}=r^2_{45}=r^2_{52}=a^2 -b^2=1$ satisfies
system \eqref{equ:complexcc}.

For the real configurations, when $a\to 0$, vertices $\textbf{1}$, $\textbf{2}$, and $\textbf{4}$ collide. The corresponding singular sequence is  a triple contact with $r_{12}, r_{14}, r_{24} \to 0$, so the corresponding diagram is a copy of Diagram 1 of Figure \ref{fig:list1} (the vertices of the triangle should be \textbf{1, 2, 4}). 
When $b\to 0$, vertices $\textbf{1}$, $\textbf{3}$, and $\textbf{5}$ collide. Similarly, the corresponding diagram is a copy of Diagram 1 of Figure \ref{fig:list1} (the vertices of the triangle should be \textbf{1, 3, 5}).

For the complex configurations, when $a\to \infty$, we have two choices: $b\sim a$ or $b\sim -a$. We get two other singular sequences with $\{z_2,z_3, z_4, z_5, w_2,w_3, w_4, w_5 \}$ all going to infinity. For instance,  if $a\to \infty$ and $b\sim a$,  
the corresponding diagram is exactly  Diagram 18 of Figure \ref{fig:list2}.

\section{Constraints  when some sub-diagrams appear} \label{sec:moreproperty}

\indent\par
We collect some useful results in this section. 
	We will use notations such as \(\Gamma_J, \Gamma_{j_1, \ldots, j_n}, L_J,\) and \(L_{j_1, \ldots, j_n}\) below. Please refer to Definition \ref{def:LI} for their meanings.

\begin{proposition}\label{Prp:sumT12} \cite{yu2021Finiteness}
	Suppose that a diagram has two $z$-circled vertices (say $\textbf{1}$ and $\textbf{2}$) which are also $z$-close,   if none of all the other vertices is $z$-close with them,  then $\Gamma_1+\Gamma_2\neq 0$ and $\overline{\Lambda}z_{12}w_{12}\sim \frac{1}{\Gamma_1+\Gamma_2}$. In particular, vertices $\textbf{1}$ and $\textbf{2}$ cannot form a $z$-stroke.
\end{proposition}

\begin{corollary}\label{Cor:sumT12} 
	Suppose that a diagram has  two $z$-circled vertices (say $\textbf{1}$ and $\textbf{2}$) which also form a $z$-stroke. If  none of all the other vertices is $z$-close with them, then $z_{12}\approx \epsilon^{-2}$,  $\Gamma_1+\Gamma_2\neq 0$, and $w_1, w_2\preceq \epsilon ^2$.  
\end{corollary}

\begin{proof}
	If they are $z$-close, by Proposition \ref{Prp:sumT12}, they cannot form a $z$-stroke, which is a contradiction. Note that    Rule III implies  
	\[\epsilon^{-2}   \succ  \Gamma_1 z_1 +\Gamma_2 z_2 = (\Gamma_1+\Gamma_2) z_1 +\Gamma_2 (z_2-z_1).   \]
	We obtain $\Gamma_1+\Gamma_2\neq 0.$
	
	Note that $z_{1j}, z_{2j}\approx  \epsilon^{-2}, j\ge3$. Then 
	\[\bar{\Lambda}\sum_{j\ge 3} \Gamma_j w_j=\sum_{j\ge 3}\frac{\Gamma_1\Gamma_j}{z_{1j}}+\sum_{j\ge3}\frac{\Gamma_2\Gamma_j}{z_{2j}}\preceq \epsilon^{2}.\]
	By the equation $\sum_{j}\Gamma_j w_j=0$, we have 
	\[ \epsilon^{2}\succeq \Gamma_1 w_1+\Gamma_2w_2=(\Gamma_1+\Gamma_2)w_1+ \Gamma_2w_{21}.  \]
	Since $w_{21}\approx \epsilon^2$, we have $w_1, w_2\preceq \epsilon ^2$.  
\end{proof}

\begin{proposition}\label{Prp:LI} \cite{yu2021Finiteness}
	Suppose that  a fully $z$-stroked sub-diagram with  vertices $\{1,..., k\}, (k\ge 3)$ exists in isolation  in a diagram, and none of its vertices  is $z$-circled, then
	\begin{equation}\notag
	L_{1...k}= \sum_{i,j \in \{1, ..., k\}, i\ne j}\Gamma_i\Gamma_j=0.
	\end{equation}
	
\end{proposition}

\begin{corollary}\label{Cor:L&Gamma} 
	Suppose a fully \(z\)-stroked sub-diagram with vertices \(K = \{1, \ldots, k\}\), \(k \ge 3\), exists in isolation in a diagram, and none of its vertices is \(z\)-circled. 
	\begin{enumerate}
		\item If there is an isolated component \(I\) of the \(w\)-diagram such that \(K \subset I\), then the \(w\)-circled vertices in \(I\) cannot be exactly \(\{1, \ldots, k\}\). 
		\item Consider  any subset of \(K\) with cardinality \((k-1)\) , say \(K_1 = \{2, \ldots, k\}\). If there is an isolated component \(I\) of the \(w\)-diagram with \(K_1 \subset I\), then the \(w\)-circled vertices in \(I\) cannot be exactly \(K_1\). 
		\item If there is a vertex outside of \(K\), say \(k+1\), such that \(\{k+1\} \cup K\) forms an isolated component of the \(w\)-diagram and these \(k+1\) vertices are fully \(w\)-stroked, then there is at least one \(w\)-circle among them.  
		\item If there are several isolated components \(\{I_j, j = 1, \ldots, s\}\) of the \(w\)-diagram with \(K \subset \cup_{j=1}^s I_j\), then the \(w\)-circled vertices in \(\cup_{j=1}^s I_j\) cannot be exactly \(\{1, \ldots, k\}\). 
	\end{enumerate}
\end{corollary}

\begin{proof}
	First, we have \(L_{K} = 0\) by Proposition \ref{Prp:LI}, and the vertices of \(K\) are all \(w\)-close by  the estimate of Proposition \ref{Estimate1}.
	
	For part (1), if the \(w\)-circled vertices in \(I\) are exactly \(\{1, \ldots, k\}\), then by Rule IV, we have \(\sum_{i \in K} \Gamma_i = 0\). This leads to a contradiction because:
	\[
	\left(\sum_{i \in K} \Gamma_i\right)^2 = \sum_{i \in K} \Gamma_i^2 + 2L_{K}.
	\]
	The proof of part (4) is similar.
	
	For part (2), if the \(w\)-circled vertices in \(I\) are exactly \(K_1 = \{2, \ldots, k\}\), then by Rule IV, we have \(\sum_{i=2}^k \Gamma_i = 0\). Therefore:
	\[
	L_{K_1} = L_{K} - \Gamma_{1}\left(\sum_{i=2}^k \Gamma_i\right) = 0,
	\]
	which again leads to a contradiction since \(\sum_{i \in K_1} \Gamma_i = 0\).
	
	For part (3), if the component \(\{k+1\} \cup K\) is fully \(w\)-stroked but has no \(w\)-circle, then Rule IV implies that \(L_K = 0\) and \(L_K + \Gamma_{k+1} \sum_{i \in K} \Gamma_i = 0\). This leads to \(\sum_{i \in K} \Gamma_i = 0\), which is a contradiction.
\end{proof}

\begin{proposition}\label{Prp:isolate-z_12}
	Suppose that a diagram has an isolated \(z\)-stroke in the \(z\)-diagram, and its two ends are \(z\)-circled. Let \(\textbf{1}\) and \(\textbf{2}\) be the ends of this \(z\)-stroke. Suppose there is no other \(z\)-circle in the diagram. Then \(\Gamma_1 + \Gamma_2 \ne 0\), and \(z_{12}\) is maximal. The diagram forces \(\Lambda = \pm 1\) or \(\pm \textbf{i}\).
	Furthermore,
	\begin{itemize}
		\item If $\Lambda = \pm 1$, we have $\sum_{j=3}^N \Gamma_j=0$;
		\item 	If $\Lambda = \pm \textbf{i}$, we have $L=0$ and $\Gamma_1\Gamma_2=L_{3...N}$.
	\end{itemize}
\end{proposition}

\begin{proof}
	
	The facts  that	$\Gamma_1+\Gamma_2\neq 0$ and $z_{12}$ is maximal  follow from Corollary \ref{Cor:sumT12}. 
	Without loss of generality, assume $z_1\sim -\Gamma_2 a\epsilon^{-2}$ and $z_2\sim \Gamma_1 a \epsilon^{-2}$, then $$z_{12}\sim (\Gamma_1 +\Gamma_2) a \epsilon^{-2},\  \frac{1}{z_{2}}-\frac{1}{z_{1}}\sim (\frac{1}{\Gamma_1}+\frac{1}{\Gamma_2})\frac{\epsilon^2}{a}. $$
	
	The  System \eqref{equ:complexcc} yields
	\begin{equation}
	\begin{array}{c}
	\label{equ:iso-z12}\overline{\Lambda} w_{12}=
	(\Gamma_1+\Gamma_2)W_{12}+ \sum_{j=3}^N \Gamma_j (\frac{1}{z_{j2}}-\frac{1}{z_{j1}}) \cr 
	\Lambda z_{2} \sim \Gamma_1 Z_{12}.
	\end{array}{}
	\end{equation}
	
	The second equation of \eqref{equ:iso-z12} implies  $w_{12} \sim \frac{\epsilon^2}{ a\Lambda}$.  Note that
	$\frac{1}{z_{j2}}-\frac{1}{z_{j1}}\sim \frac{1}{z_{2}}-\frac{1}{z_{1}}$
	for all $j >2$ and that  $W_{12}=\frac{1}{z_{12}}$. The first equation of
	\eqref{equ:iso-z12} implies 
	\begin{equation}\label{equ:iso-z12-1}
	{\overline{\Lambda}}/{\Lambda}=1+ \sum_{j=3}^N \Gamma_j  (\frac{1}{\Gamma_1}+\frac{1}{\Gamma_2}).
	\end{equation}
	It follows that $\Lambda = \pm 1$ or $ \pm \textbf{i}$.
	
	If $\Lambda = \pm 1$, we have
	\[ 0= \sum_{j=3}^N \Gamma_j  (\frac{1}{\Gamma_1}+\frac{1}{\Gamma_2}),\  \Rightarrow\  \sum_{j=3}^N \Gamma_j =0. \]
	
	If $\Lambda = \pm \textbf{i}$, we obtain 
	\[ -2= \sum_{j=3}^N \Gamma_j  (\frac{1}{\Gamma_1}+\frac{1}{\Gamma_2}),\  L=0, \ \Rightarrow L=0, \ \Gamma_1 \Gamma_2=L_{3...N}.  \]
\end{proof}

Similarly, we have the following result.

\begin{proposition}\label{Prp:isolate-z_123}
	Suppose that a diagram has an isolated triangle of \(z\)-strokes in the \(z\)-diagram, where two of the vertices of the triangle are \(z\)-circled. Let \(\textbf{2}\) and \(\textbf{3}\) be the two \(z\)-circled vertices and \(\textbf{1}\) the other vertex. Suppose there is no other \(z\)-circle in the diagram. Then \(\Gamma_2 + \Gamma_3 \ne 0\), and \(z_{23}\) is maximal. The diagram forces \(\Lambda = \pm 1\) or \(\pm \textbf{i}\).  Furthermore,
	\begin{itemize}
		\item If $\Lambda = \pm 1$, we have $\sum_{j=4}^N \Gamma_j=0$;
		\item 	If $\Lambda = \pm \textbf{i}$, we have $L=0$ and $L_{123}=L_{4...N} + \Gamma_1( \sum _{j=4}^N \Gamma_j)$.
	\end{itemize}
\end{proposition}

\begin{proof}
	The facts that 	$\Gamma_2+\Gamma_3\ne 0$ and $z_{23}$ is maximal  follow from Corollary \ref{Cor:sumT12}.  Note that
	\[  \Gamma_2 z_2 +\Gamma_3 z_3 \prec \epsilon^{-2}, \ \Lambda  z_1 \sim \Gamma_2 Z_{21} +\Gamma_3 Z_{31} \prec \epsilon^{-2}.  \]
	Without loss of generality, assume
	\[ z_2\sim -\Gamma_3 a\epsilon^{-2}, \ z_3\sim \Gamma_2 a \epsilon^{-2}, \ Z_{21} \sim -\Gamma_3 b\epsilon^{-2}, \ Z_{31}\sim \Gamma_2 b \epsilon^{-2}.  \]
	Then
	\begin{eqnarray*}
		z_{23}\sim (\Gamma_2 +\Gamma_3) a \epsilon^{-2},\  \frac{1}{z_{3}}-\frac{1}{z_{2}}\sim (\frac{1}{\Gamma_2}+\frac{1}{\Gamma_3})\frac{\epsilon^2}{a}, \cr
		Z_{23}=\frac{1}{w_{23}}= \frac{1}{1/Z_{21} +1/Z_{13}} \sim -b \frac{\Gamma_2\Gamma_3}{(\Gamma_2+\Gamma_3)}  \epsilon^{-2}.
	\end{eqnarray*}
	Then similar to the above case, we have
	\begin{eqnarray*}
		\overline{\Lambda} w_{23}\sim
		(\Gamma_2+\Gamma_3)W_{23}+ \sum_{j\ne 2, 3}^N \Gamma_j (\frac{1}{z_{3}}-\frac{1}{z_{2}}) \cr 
		\Lambda z_{23} \sim  (\Gamma_2+\Gamma_3) Z_{23} +\Gamma_1 (Z_{13}-Z_{12}).
	\end{eqnarray*}
	
	Short computation reduces the two equations to
	\begin{eqnarray*}
		-\overline{\Lambda} \frac{a}{b}= \frac{\Gamma_2\Gamma_3}{\Gamma_2+\Gamma_3} (1+\sum_{j\ne 2, 3} \Gamma_j \frac{\Gamma_2+\Gamma_3} {\Gamma_2\Gamma_3}), \cr
		-\Lambda \frac{a}{b}= \frac{L_{123}}{\Gamma_2+\Gamma_3}.
	\end{eqnarray*}
	Then we obtain
	\[   \frac{\overline{\Lambda}}{\Lambda}  L_{123}= \Gamma_2\Gamma_3+ (\Gamma_2+\Gamma_3) \sum_{j\ne 2, 3} \Gamma_j.   \]
	It follows that $\Lambda = \pm 1$ or $ \pm \textbf{i}$.
	
	If $\Lambda = \pm 1$, we have  $\sum_{j=4}^N \Gamma_j=0$.
	If $\Lambda = \pm \textbf{i}$, we have $L=0$ and
	$$L=0, \  -L_{123}= \Gamma_2\Gamma_3+ (\Gamma_2+\Gamma_3) \sum_{j\ne 2, 3} \Gamma_j,$$
	which is equivalent to
	$L_{123}=L_{4...N} + \Gamma_1( \sum _{j=4}^N \Gamma_j), \ L=0. $
\end{proof}

\begin{proposition}\label{Prp:triangle} 
	Assume there is a triangle with vertices \(\textbf{1}, \textbf{2}, \textbf{3}\) that is fully \(z\)- and \(w\)-stroked, and fully \(z\)- and \(w\)-circled. Moreover, assume that  the triangle is isolated in the \(z\)-diagram. Then there must exist  some \(k > 3\) such that \(z_{k1} \preceq 1\).
\end{proposition}

\begin{proof}
	By Proposition \ref{Estimate1} and Rule IV, we have  \[z_1\sim z_2\sim z_3, \ w_1\sim w_2\sim w_3, \ \Gamma_{1}+\Gamma_{2}+\Gamma_{3}=0.\]
	Suppose that it holds  $z_{k1}\succ 1$ for all $k>3$. Then  $\frac{1}{z_{kj}}- \frac{1}{z_{k1}}=\frac{z_{1j}}{z_{kj}z_{k1}} \prec \epsilon^2$ for all $k>3, 1\le j\le 3$, and so
	\[\bar{\Lambda}\sum_{j=1}^{3}\Gamma_j w_j=\sum_{k\ge 4} \sum_{j=1}^{3}\frac{\Gamma_k\Gamma_j}{z_{kj}}=\sum_{k\ge 4} \sum_{j=1}^{3}\Gamma_k\Gamma_j  (\frac{1}{z_{k1}}+\frac{1}{z_{kj}}- \frac{1}{z_{k1}} ) \prec \epsilon^{2}.\]

	By the fact that  $w_{12}, w_{13}, w_{23}\approx \epsilon^{2}$, the equations 
	\[\sum_{j=1}^{3}\Gamma_j w_j=\Gamma_{2}w_{12}+\Gamma_{3}w_{13}=\Gamma_{1}w_{21}+\Gamma_{3}w_{23}=\Gamma_{1}w_{31}+\Gamma_{2}w_{32} \prec \epsilon^2\]
	imply that 
	\begin{equation}\label{equ:triangle1}
	\frac{w_{12}}{\Gamma_{3}}\sim\frac{w_{23}}{\Gamma_{1}}\sim\frac{w_{31}}{\Gamma_{2}}\approx \epsilon^{2}.
	\end{equation}
	By the isolation of this triangle in the $z$-diagram, it holds that 
	\begin{equation}\label{equ:triangle2}
	\Lambda z_{1}\sim \frac{\Gamma_2}{w_{21}}+\frac{\Gamma_3}{w_{31}}, ~\Lambda z_{2}\sim \frac{\Gamma_1}{w_{12}}+\frac{\Gamma_3}{w_{32}},~\Lambda z_{3}\sim \frac{\Gamma_1}{w_{13}}+\frac{\Gamma_2}{w_{23}}.
	\end{equation}
	Since $z_1\sim z_2\sim z_3$, the equations \eqref{equ:triangle1} and \eqref{equ:triangle2} lead to 
	\[\frac{\Gamma_1}{\Gamma_{2}}-\frac{\Gamma_2}{\Gamma_{1}}=\frac{\Gamma_2}{\Gamma_{3}}-\frac{\Gamma_3}{\Gamma_{2}}=\frac{\Gamma_3}{\Gamma_{1}}-\frac{\Gamma_1}{\Gamma_{3}}.\] 
	This  contradicts with $\Gamma_{1}+\Gamma_{2}+\Gamma_{3}=0$. 
\end{proof}

Similarly, we have the following result.
\begin{proposition}\label{Prp:triangle2} 
	Suppose that a diagram has an isolated triangle of \(z\)-strokes in the \(z\)-diagram, where all three vertices, say \(\textbf{1}, \textbf{2}, \textbf{3}\), are \(z\)-circled. If \(z_1 \sim z_2 \sim z_3\), then there exists some \(k > 3\) such that \(z_{k1} \prec \epsilon^{-2}\). 
\end{proposition}

\begin{proof}
	Suppose that 	$z_1\sim z_2\sim z_3\approx \epsilon^{-2}$. By Proposition \ref{Estimate1} and Rule IV, we have $\Gamma_{1}+\Gamma_{2}+\Gamma_{3}=0$.
	Suppose that it holds  $z_{k1} \approx \epsilon^{-2}$ for all $k>3$. Then  $\frac{1}{z_{kj}}- \frac{1}{z_{k1}}=\frac{z_{1j}}{z_{kj}z_{k1}} \prec \epsilon^2$ for all $k>3, 1\le j\le 3$. Similar to the argument of the above result, we have  $\bar{\Lambda}\sum_{j=1}^{3}\Gamma_j w_j \prec \epsilon^{2}$, 
	\[  \frac{w_{12}}{\Gamma_{3}}\sim\frac{w_{23}}{\Gamma_{1}}\sim\frac{w_{31}}{\Gamma_{2}}\approx \epsilon^{2},\]
	\begin{equation}\notag
	\Lambda z_{1}\sim \frac{\Gamma_2}{w_{21}}+\frac{\Gamma_3}{w_{31}}, ~\Lambda z_{2}\sim \frac{\Gamma_1}{w_{12}}+\frac{\Gamma_3}{w_{32}},~\Lambda z_{3}\sim \frac{\Gamma_1}{w_{13}}+\frac{\Gamma_2}{w_{23}}, 
	\end{equation}
	and 
	$\frac{\Gamma_1}{\Gamma_{2}}-\frac{\Gamma_2}{\Gamma_{1}}=\frac{\Gamma_2}{\Gamma_{3}}-\frac{\Gamma_3}{\Gamma_{2}}=\frac{\Gamma_3}{\Gamma_{1}}-\frac{\Gamma_1}{\Gamma_{3}}.$
	This  contradicts with $\Gamma_{1}+\Gamma_{2}+\Gamma_{3}=0$. 
\end{proof}

\begin{proposition}\label{Prp:dumbbell} 
	Assume that vertices \textbf{1} and \textbf{2} are both $z$- and $w$-circled and connected by a $zw$-edge, and the sub-diagram formed by the two vertices  is isolated in the $z$-diagram. Assume that 
	vertices \textbf{3} and \textbf{4} are also both $z$- and $w$-circled and connected by a $zw$-edge,  and  is isolated in the $z$-diagram. 
	Then, there must exist some $k>4$ such that at least one among  $z_{k1}, w_{k1}, z_{k3}, w_{k3}$ is bounded (i.e., $\preceq 1$). 
\end{proposition}

\begin{proof}
	By Proposition \ref{Estimate1} and Rule IV, we have  \[z_1\sim z_2,  z_3\sim z_4, \ w_1\sim w_2, w_3 \sim w_4, \ \Gamma_{1}+\Gamma_{2}=0, \Gamma_3+\Gamma_{4}=0.\]
	Suppose that it holds that $w_{k1}\succ 1$ for all $k>4$. Then  $\frac{1}{w_{k2}}- \frac{1}{w_{k1}}=\frac{w_{12}}{w_{k2}w_{k1}} \prec \epsilon^2$ for all $k>4$.  
	Note that $z_{12}\approx \epsilon^2$, and 
	\[\Lambda z_{12}= (\Gamma_1+\Gamma_2)Z_{12}+ \Gamma_3( \frac{w_{21}}{w_{32}w_{31}} -\frac{w_{21}}{w_{42}w_{41}}) + \sum_{k>4}\Gamma_k(\frac{1}{w_{k2}}- \frac{1}{w_{k1}} ).     \]
	We conclude that  
	\[\frac{w_{21}}{w_{31}w_{32}}\approx  \frac{w_{21}}{w_{41}w_{42}}\succeq \epsilon^{2}\Rightarrow w_{31}\preceq 1\Rightarrow w_1\sim w_2\sim w_3\sim w_4.\]
	Similarly, we have\[z_1\sim z_2\sim z_3\sim z_4.\]

	Note that 
	\[\overline{\Lambda}\sum_{j=1}^{4}\Gamma_j w_j=\sum_{k> 4} \sum_{j=1}^{4}\frac{\Gamma_k\Gamma_j}{z_{kj}}=\sum_{k> 4} \Gamma_k \Gamma_1 (\frac{1}{z_{k1}}- \frac{1}{z_{k2}})  + \sum_{k> 4} \Gamma_k \Gamma_3 (\frac{1}{z_{k3}}- \frac{1}{z_{k4}})  \prec \epsilon^{2}. \]
	Then the equation $\Gamma_2 w_{12}+\Gamma_4 w_{34}=\sum_{j=1}^{4}\Gamma_j w_j$ leads to
	\[\Gamma_2 w_{12}\sim -\Gamma_4 w_{34}, \ \text{or} \ \Gamma_2 Z_{34}\sim -\Gamma_4 Z_{12} .\] 
	On the other hand, the isolation of the two segments implies 
	\[\Lambda z_2\sim \Gamma_1 Z_{12}, \ \Lambda z_4\sim \Gamma_3 Z_{34}, \Rightarrow \Gamma_1 Z_{12}\sim \Gamma_3 Z_{34}.\]
	As a result, we have
	\[\Gamma_1\Gamma_2=-\Gamma_3\Gamma_4, \ \text{or} \ \Gamma_1^2+\Gamma_3^2=0,\] 
	which is a contradiction.
	
\end{proof}

\begin{proposition}\label{Prp:quadrilateral} 
	Assume that there is a quadrilateral  with  vertices \textbf{1,  2, 3, 4}, that is fully $z$- and $w$-stroked, and fully  $w$-circled.   Moreover, the quadrilateral is isolated  in the $w$-diagram. 
	Then, there must exist some  $k>4$ such that $w_{k1} \preceq 1$. 
\end{proposition}

\begin{proof}
	We establish the result by contradiction. 
	Rule IV implies that   $\sum_{j=1}^{4}\Gamma_j=0.$ 	Suppose that it holds that $w_{k1}\succ 1$ for all $k>4$. Then $\frac{1}{w_{kj}}- \frac{1}{w_{k1}}\prec  \epsilon^{2}$ for $k>4, j\le4$, so
	\[\Lambda\sum_{j=1}^{4}\Gamma_j z_j=\sum_{k>4}\Gamma_k \sum_{j=1}^{4}\frac{\Gamma_j}{w_{kj}}=\sum_{k>4}\Gamma_k  \sum_{j=1}^{4} \Gamma_j  (\frac{1}{w_{k1}}+\frac{1}{w_{kj}}- \frac{1}{w_{k1}} )   \prec \epsilon^{2}.\]
	Then  \[  \epsilon^2 \succ \sum_{j=1}^{4}\Gamma_j z_j=\sum_{j=2, 3, 4}\Gamma_j z_{1j}\Rightarrow \Gamma_2 z_{12}\sim -\Gamma_3 z_{13}-\Gamma_4 z_{14}\approx \epsilon^{2}.  \]
	Set $z_{13}\sim a \epsilon^2, z_{14}\sim b \epsilon^2$, where $a\ne b$ are some nonzero constants. Then 
	\begin{align*}
		&z_{12}\sim -\frac{\Gamma_3 a+\Gamma_4 b}{\Gamma_2} \epsilon^2,    & z_{23}\sim \frac{a (\Gamma_2+\Gamma_3)+\Gamma_4 b}{\Gamma_2} \epsilon^2,\\
		&z_{24}\sim \frac{\Gamma_3 a+b (\Gamma_2+\Gamma_4)}{\Gamma_2} \epsilon^2,    &z_{34}\sim (b-a) \epsilon^2. 
	\end{align*}
	Since $w_1\sim w_2\sim w_3\sim w_4$, we set $\bar{\Lambda}w_k\sim \frac{1}{c\epsilon^2}, k=1, 2, 3, 4$.
	Substituting those into the system
	\[\bar{\Lambda}w_k\sim \sum_{j\ne k, j=1}^{4}\frac{\Gamma_j}{z_{jk}}, \ k=1, 2,3,4,\]
	which is from the isolation of the quadrilateral in $w$-diagram. 
	We obtain four homogeneous  polynomials of the three variables $a, b, c$. Thus, we set $c=1$, and obtain the following four polynomials of the five variables $a,b,  \Gamma_1, \Gamma_{3}, \Gamma_{4}$, 
	\begin{align*}
		&a^2 (-\Gamma_3) (b+\Gamma_4)+a b \left(-\Gamma_4 (b-2 \Gamma_3)+\Gamma_1^2+2 \Gamma_1 (\Gamma_3+\Gamma_4)\right)-b^2 \Gamma_3 \Gamma_4=0,\\
		&a^3 \Gamma_3^2 (\Gamma_1+\Gamma_4)+a^2 \Gamma_3 \left(-b \left(\Gamma_1^2+\Gamma_1 \Gamma_3+\Gamma_4 (2 \Gamma_3-\Gamma_4)\right)-(\Gamma_1-\Gamma_3+\Gamma_4) (\Gamma_1+\Gamma_3+\Gamma_4)^2\right)\\
		&-a b \left(\Gamma_1^2 \Gamma_4 (b-4 \Gamma_3)+\Gamma_1 \left(\Gamma_4^2 (b+2 \Gamma_4)+2 \Gamma_3^3-2 \Gamma_3^2 \Gamma_4-2 \Gamma_3 \Gamma_4^2\right)-\Gamma_3^2 \Gamma_4 (b+2 \Gamma_4)\right)\\
		&-a b \left(2 b \Gamma_3 \Gamma_4^2-\Gamma_1^4-2 \Gamma_1^3 (\Gamma_3+\Gamma_4)+\Gamma_3^4+\Gamma_4^4\right)\\
		&+b^2 \Gamma_4 \left(\Gamma_1 \left(\Gamma_4 (b+\Gamma_4)-3 \Gamma_3^2-2 \Gamma_3 \Gamma_4\right)+\Gamma_3 \Gamma_4 (b+\Gamma_4)-\Gamma_1^3-\Gamma_1^2 (3 \Gamma_3+\Gamma_4)-\Gamma_3^3-\Gamma_3^2 \Gamma_4+\Gamma_4^3\right)=0,\\
		&a^3 (\Gamma_1+\Gamma_4)+a^2 (\Gamma_3 (2 \Gamma_1+\Gamma_3+2 \Gamma_4)-b (\Gamma_1+2 \Gamma_4))+a b \left(b \Gamma_4-2 \Gamma_1 \Gamma_3-\Gamma_3^2-2 \Gamma_3 \Gamma_4\right)-b^2 \Gamma_1 \Gamma_4=0,\\
		&a^2 \Gamma_3 (b-\Gamma_1)-a b (b (\Gamma_1+2 \Gamma_3)+\Gamma_4 (2 \Gamma_1+2 \Gamma_3+\Gamma_4))+b^2 (b (\Gamma_1+\Gamma_3)+\Gamma_4 (2 \Gamma_1+2 \Gamma_3+\Gamma_4))=0. 
	\end{align*}
	Tedious but standard computation, such as calculating the Gr\"{o}bner basis,   yields 
	\[  b^5 (\Gamma_1+\Gamma_3+\Gamma_4) \left(\Gamma_1^2+\Gamma_1 \Gamma_3+\Gamma_1 \Gamma_4+\Gamma_3^2+\Gamma_3 \Gamma_4+\Gamma_4^2\right)=0. \]
	It is a contradiction since $b\ne 0, \Gamma_1+\Gamma_3+\Gamma_4=-\Gamma_2\ne 0$ and 
	\[\Gamma_1^2+\Gamma_1 \Gamma_3+\Gamma_1 \Gamma_4+\Gamma_3^2+\Gamma_3 \Gamma_4+\Gamma_4^2=\frac{1}{2}( \Gamma_1^2+\Gamma_2^2+\Gamma_3^2+\Gamma_4^2)\ne 0.\]
\end{proof}

\section{Construction of  the 5-vortex diagrams} \label{sec:matrixrules}

\indent\par
From now on, we focus on  the planar 5-vortex problem.  In this section,  we identify all problematic  diagrams for the  5-vortex central configurations. 
 We adopt  the approach  used  by Albouy and Kaloshin in  the $N$-body  problem \cite{Albouy2012Finiteness}, which involves analyzing diagrams according to the maximal number of strokes emanating from a two-colored vertex.
During the analysis of all the possibilities we rule some of them out immediately. The  ones we cannot exclude without further consideration  are
 collected into a list of 31 diagrams, shown in Figures \ref{fig:list1} and \ref{fig:list2} in Section \ref{sec:diagram&constraints}.

We call a \emph{two-colored vertex} of the diagram a vertex which connects at least
a stroke of $z$-color with at least a stroke of $w$-color.  The number of strokes from
a two-colored vertex is at least 2 and at most 8. Given a diagram, we
define $C$ as \emph{the maximal number of strokes from a two-colored vertex}. We use
this number to classify all possible diagrams.

Recall that the $z$-diagram indicates the maximal terms among a finite set
of terms. It is nonempty. If there is a circle, there is an edge of the same
color emanating from it. So there is at least a $z$-stroke, and similarly, at least a
$w$-stroke.

\begin{remark}
	We developed a symbolic computation algorithm to determine the diagrams for the \(N\)-vortex problem \cite{YuZhuarxiv}, inspired by the work of Chang and Chen, who designed symbolic computation algorithms to implement the singular sequence method for the \(N\)-body problem in celestial mechanics \cite{2023Chen-1}. The core idea is to represent each two-colored diagram as an \(N \times 2N\) binary matrix. The rules from Section \ref{sec:rule} and the results from Section \ref{sec:moreproperty} are then translated into conditions on these binary matrices. Using this algorithm, we can efficiently filter out invalid diagrams. This approach is particularly useful when applying the singular sequence method to the finiteness problem for \(N \ge 6\), as it significantly reduces the required manual work.  For $N=5$,  the algorithm outputs 31 diagrams, matching the list shown in Figures  \ref{fig:list1} and \ref{fig:list2} in Section \ref{sec:diagram&constraints}.  However, we emphasize that all computations in this section were carried out by hand; the algorithm was not used to obtain any of these diagrams.
\end{remark}

\subsection{No two-colored vertex }

There is at least one isolated edge, which is not a $zw$-edge. Let us say it is a $z$-edge. The complement has 3 bodies. There three can have one or  three $w$-edges according to Rule VI.

For one $w$-edge, the attached bodies have to be $w$-circled by Rule I. This is the first diagram in Figure \ref{fig:C=0}.

For three $w$-edges, the three edges form a triangle. There are three possibilities for the number of  $w$-circled vertices: it is either zero, or two or three (one is not possible by Rule III.)  They constitute the last three diagrams in Figure \ref{fig:C=0}.

Hence, we have \emph{four} possible diagrams, as shown in Figure \ref{fig:C=0}. 

\begin{figure}[h!]
	
	\centering
\subfigure{	\includegraphics[width=0.2\textwidth]{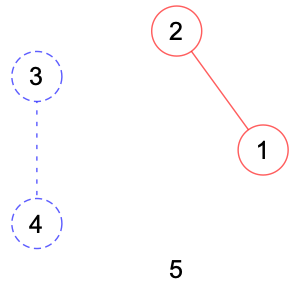}}\hspace{8pt}
	\subfigure{	\includegraphics[width=0.2\textwidth]{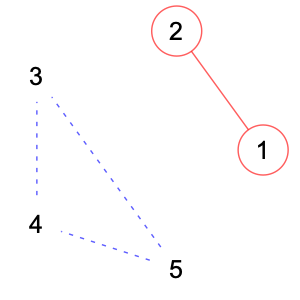}}\hspace{8pt}
	\subfigure{	\includegraphics[width=0.2\textwidth]{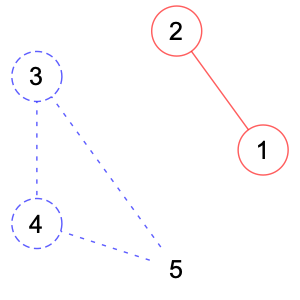}}\hspace{8pt}
	\subfigure{	\includegraphics[width=0.2\textwidth]{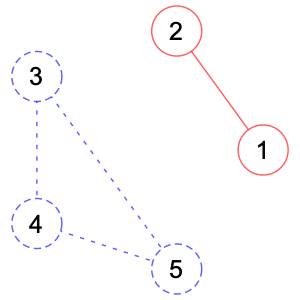}}
	
	\caption{Four possible diagram for no two-colored vertex. They correspond to Diagram  15 of Figure \ref{fig:list2},  Diagram   1 of Figure \ref{fig:list1}, Diagram  6 of Figure \ref{fig:list1}, and Diagram  16 of Figure \ref{fig:list2} respectively.    }
	\label{fig:C=0}
\end{figure}

\subsection{$C=2$}

There are two cases: a $zw$-edge exists or not.

If it is present,  it is isolated. Let us say, vertex $\textbf{1}$ and vertex $\textbf{2}$ are connected by one $zw$-edge.   Note that  there must be  both  $z$ and $w$-circle  among vertices $\textbf{3}$, $\textbf{4}$, and $\textbf{5}$. If none of the three vertices is $z$-circled, we have $\Gamma_1+\Gamma_2=0$ by Rule IV. On the other hand, since vertices $\textbf{3}$, $\textbf{4}$, and $\textbf{5}$ are not $z$-circled, they are not $z$-close to vertex $\textbf{1}$ and vertex $\textbf{2}$ . Then
Proposition \ref{Prp:sumT12} implies that $\Gamma_1+\Gamma_2\ne0$.  This is a contradiction.

Then Rule I implies that there is at least one $z$-stroke and one $w$-stroke among the cluster of vertices $\textbf{3}$, $\textbf{4}$, and $\textbf{5}$. There are two possibilities: whether there is another $zw$-edge or not.

If another $zw$-edge is present, then it is again isolated. This is the first diagram in Figure \ref{fig:C=2}. Note that $w_{15}\approx \epsilon^{-2}$, which  contradicts with Proposition \ref{Prp:dumbbell}, thus impossible.

If another $zw$-edge is not present, there is at least one edge in both color.
By the circling method, the adjacent vertex is $z$- and $w$-circled.  By the Estimate, the $z$-edge implies the two attached vertices are $w$-close. Then the two ends are both $w$-circled by Rule II. Thus, all three vertices are $z$- and $w$-circled.  Then there are  more strokes in the diagram. This is a contradiction.

If there is no $zw$-edge,  there are adjacent $z$-edges and $w$-edges. From any such adjacency there is no other edge.  Suppose that  vertex $\textbf{1}$ connects with vertex $\textbf{4}$ by $w$-edges and connects with $\textbf{2}$ by $z$-edges. The circling method implies that $\textbf{1}$ is $z$- and $w$-circled, 2 is $w$-circled and 4 is $z$-circled. The color of  $\textbf{2}$ and $\textbf{4}$ forces the color of edges from the circle. If one of the two new edges completes the triangle with vertices $\textbf{1, 2, 4}$, then Rule VI implies that $C>2$, a contradiction. If the two new edges
go to the same vertex, we get the diagram corresponding to Roberts' continuum at infinity, shown as the second  in Figure \ref{fig:C=2}.

If the two edges go to the different  vertices,  the circling method and Rule I demand a cycle with alternating colors, which is impossible since the cycle has  five edges.

Hence,  there is only \emph{one}  possible  diagram,  the second one in Figure \ref{fig:C=2}.
\begin{figure}[!h]
	
	\centering
		\centering
	\subfigure{	\includegraphics[width=0.2\textwidth]{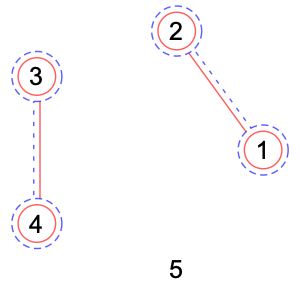}}\hspace{12pt}
	\subfigure{	\includegraphics[width=0.2\textwidth]{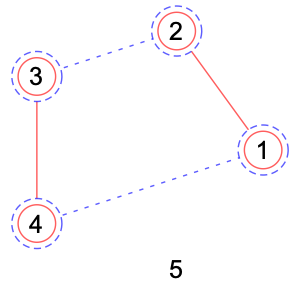}}
	
	\caption{Two diagrams for $C=2$. The first one has been excluded. The second one corresponds to Diagram  18 of Figure \ref{fig:list2}.   }
	\label{fig:C=2}
\end{figure}

\subsection{$C=3$}

Consider a two-colored vertex  with three strokes.  There are two cases: in the first, it is like  vertex $\textbf{1}$ in Figure \ref{fig:C=3};  in  the second,  it connects a single stroke to a $zw$-edge.

We start with the   first  case.  Let us say vertex $\textbf{1}$ connects with vertex $\textbf{2}$ and vertex $\textbf{3}$ by   $z$-edges, and  connects with  vertex $\textbf{4}$ by  a $w$-edge.  There is a $z_{23}$-stroke by Rule VI. The circling method implies that the vertices  $\textbf{1}$, $\textbf{2}$ and vertex $\textbf{3}$ are all $w$-circled, see Figure \ref{fig:C=3}.  Then there is $w$-stroke emanating from $\textbf{2}$ and vertex $\textbf{3}$.  The $w$-stroke may go to vertex $\textbf{4}$, vertex $\textbf{5}$, or it is a  $w_{23}$-stroke.

If  one $w$-stroke goes from vertex $\textbf{2}$ to vertex  $\textbf{4}$, then there is extra $w_{12}$-stroke by Rule VI, which contradict with $C=3$. If all two $w$-strokes go to vertex $\textbf{5}$, then Rule VI implies the existence of $w_{23}$-stroke. This is again a contradiction with $C=3$.

If the $w$-strokes emanating from $\textbf{2}$ and vertex $\textbf{3}$  are just the  
$w_{23}$-stroke,  then we have a $zw$-edge between $\textbf{2}$, and $\textbf{3}$. Then it is not necessary to discuss the  second case.   Then,  we consider the vertex $\textbf{5}$.  It is connected with the previous four vertices or isolated.

If the diagram is connected, vertex $\textbf{5}$   can only connects with vertex $\textbf{4}$  by a $z$-edge (other cases is not possible by Rule VI). Then the circling method implies that vertex $\textbf{5}$   is $w$-circled. Then there is $w$-stroke emanating from $\textbf{5}$. This is a contradiction.

If vertex $\textbf{5}$   is isolated. Then  the circling method implies that all vertices except vertex $\textbf{5}$   are $w$-circled. Only $\textbf{2}$ and vertex $\textbf{3}$ can be $z$-circled, and they are both $z$-circled or both not $z$-circled by Rule IV, see Figure \ref{fig:C=3}.

If both vertex $\textbf{2}$ and vertex $\textbf{3}$ are not $z$-circled, then we have $\Gamma_2+\Gamma_3=0$ and $L_{123}=0$ by Rule IV and Proposition \ref{Prp:LI}. This is a contradiction since
\[  L_{123}= \Gamma_1 (\Gamma_2+\Gamma_3)+ \Gamma_2\Gamma_3.  \]

If both vertex $\textbf{2}$ and vertex $\textbf{3}$ are $z$-circled, then Rule IV implies that
$\Gamma_2+\Gamma_3=0$. On the other hand, Corollary \ref{Cor:sumT12}  implies that
$\Gamma_2+\Gamma_3\ne0$. This is a contradiction.

Thus, the two   diagrams in Figure \ref{fig:C=3} are both excluded. There is  no possible  diagram. 

\begin{figure}[!h]
	
	\centering
	\includegraphics[width=0.4\textwidth]{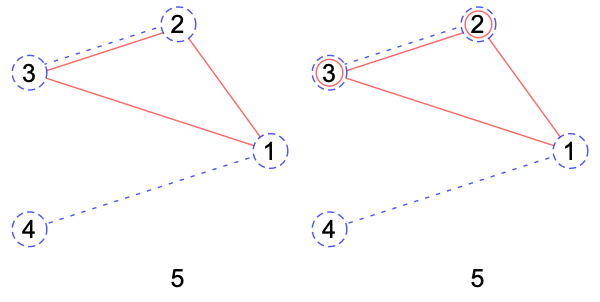}
	
	\caption{Two diagrams for $C=3$. Both have been excluded.}
	\label{fig:C=3}	
\end{figure}


\subsection{$C=4$}

There are five cases: two $zw$-edges,  only one $zw$-edge  and one  edge of each color,    only one $zw$-edge  and two edge of the same color, one  $z$-edge and three $w$-edges, or   two edges of each color emanating from the same vertex.

\subsubsection{} Suppose that there are two $zw$-edges emanating from, e.g.,  vertex  $\textbf{1}$ as on the first diagram in Figure \ref{fig:C=41}.

We get a fully $zw$-edged triangle by Rule VI. This triangle is isolated since
$C=4$.  Since  vertices $\textbf{1}$,  $\textbf{2}$, and $\textbf{3}$ are  $z$-close and $w$-close, if one of them is circled in some color, all of them will be circled in  the same color. Thus,  the first three vertices may be all $z$-circled,  all  $z$-and $w$-circled, or all not circled.

If  vertices $\textbf{1}$,  $\textbf{2}$, and $\textbf{3}$ are  $z$-circled but not $w$-circled, we have $\Gamma_1+\Gamma_2+\Gamma_3=0$ and $L_{123}=0$ by Rule IV and Proposition \ref{Prp:LI}. This is a contradiction since $(\Gamma_1+\Gamma_2+\Gamma_3)^2-2L_{123}\ne 0$. Then the first three vertices can only be all  $z$-and $w$-circled, or all not circled.

The other two vertices can be disconnected, connected by one , e.g.,  $z$-edge, or by one $zw$-edge. Then there are six possibilities, according to whether the first  three vertices are   all  $z$-and $w$-circled, or all not circled, and the connection between the other two vertices. 

Suppose that vertex $\textbf{4}$  and vertex $\textbf{5}$   are connected by one $zw$-edge, and the first three vertices are all not circled. 
Then we have $\Gamma_4+\Gamma_5=0$ by Rule IV. On the other hand, 
Corollary \ref{Cor:sumT12}  implies that $\Gamma_4+\Gamma_5\ne0$.  This is a contradiction.

Then we have the five possibilities, see Figure \ref{fig:C=41}.   Note that Proposition \ref{Prp:triangle}  can further exclude the second and fourth one.

\begin{figure}[!h]
	
	\centering
		\centering
	\subfigure{	\includegraphics[width=0.18\textwidth]{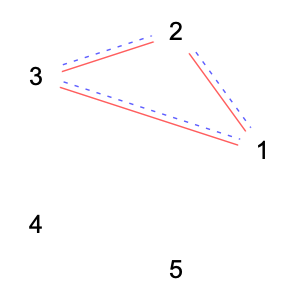}}\hspace{1pt}
	\subfigure{	\includegraphics[width=0.18\textwidth]{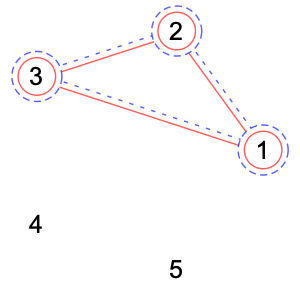}}\hspace{1pt}
	\subfigure{	\includegraphics[width=0.18\textwidth]{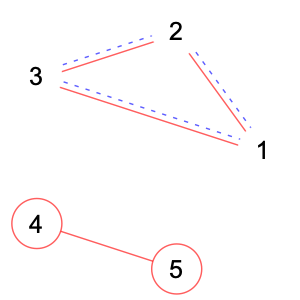}}\hspace{1pt}
	\subfigure{	\includegraphics[width=0.18\textwidth]{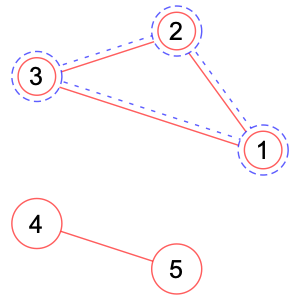}}\hspace{1pt}
	\subfigure{	\includegraphics[width=0.18\textwidth]{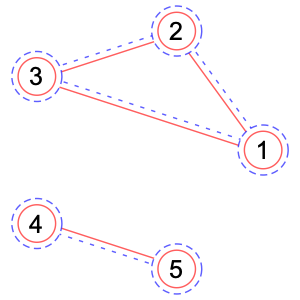}}

	\caption{Five diagrams for $C=4$, two $zw$-edges. The second and fourth one has been excluded. The first, third and fifth one correspond to Diagram  1 of Figure \ref{fig:list2},  Diagram  2 of Figure \ref{fig:list1}, and Diagram  20 of Figure \ref{fig:list2} respectively.    }
	\label{fig:C=41}
\end{figure}
Hence,  there are \emph{three}  possible  diagrams,  the first, third and fifth one in Figure \ref{fig:C=41}.

\subsubsection{}
Suppose that there is one $zw$-edge and one edge of each color emanating from vertex $\textbf{1}$, as in the first digram of Figure \ref{fig:C=42}. We complete the triangles  by Rule VI.  Note that no more strokes can emanating from vertex $\textbf{1}$ and vertex $\textbf{4}$  since $C=4$. If there are more strokes from vertex $\textbf{2}$,  it can not goes to vertex $\textbf{3}$, since it implies one more stroke emanating from vertex $\textbf{1}$.  Similarly, there is no $z_{35}$-stroke or $w_{25}$-stroke. Then between vertex $\textbf{5}$   and the first four vertices, there can have no edge, one
$z_{25}$-stroke, or one $z_{25}$-stroke and one $w_{35}$-stroke.

For the disconnected diagram,  vertex $\textbf{1}$ and vertex $\textbf{4}$  can not be circled. Otherwise, the circling method implies  either vertex $\textbf{2}$ is be $z$-circled or 3 is $w$-circled,   there should be stroke emanating from vertex $\textbf{2}$ or vertex $\textbf{3}$ . This is a contradiction.  Then vertex $\textbf{2}$ can not be circled, otherwise, vertex $\textbf{3}$  is of the same color by Rule IV. This is a contradiction. Hence, there is no circle in the diagram and this is the first diagram in Figure \ref{fig:C=42}.

If there is only $z_{25}$-stroke, then the circling method implies that  vertices $\textbf{1}$,  $\textbf{2}$, $\textbf{4}$, and $\textbf{5}$ are $z$-circled. Note that vertex $\textbf{3}$  and vertex $\textbf{5}$   can not be $w$-circled, otherwise there are extra $w$-strokes. Then  vertices $\textbf{1}$,  $\textbf{2}$, and $\textbf{4}$ are not $w$-circled
by the circling method.  Note that vertex $\textbf{3}$ must be $z$-circled, otherwise, we have $L_{124}=0$ and $\Gamma_1+\Gamma_4=0$. This is a contradiction.
Then we have the second diagram in Figure \ref{fig:C=42}.

If there are  $z_{25}$-stroke and  $w_{35}$-stroke. Then vertex $\textbf{5}$   is  $z$-and $w$-circled. By the circling method, all vertices are   $z$-and $w$-circled. This the third diagram in Figure  \ref{fig:C=42}.
\begin{figure}[!h]
	
	\centering
	\includegraphics[width=0.8\textwidth]{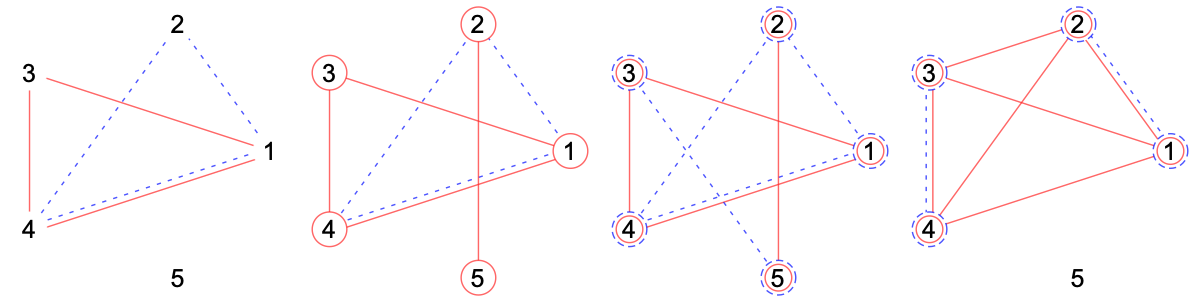}
	
	\caption{Four diagrams for $C=4$, one $zw$-edge. The  fourth one has been excluded. The first three correspond to Diagram 2 of Figure \ref{fig:list2},  Diagram  8 of Figure \ref{fig:list2},  and Diagram  19 of Figure \ref{fig:list2} respectively.    }
	\label{fig:C=42}
\end{figure}

Hence,  there are \emph{three}  possible  diagrams,  the first three  in Figure \ref{fig:C=42}.

\subsubsection{}

Suppose that there are one $zw$-edges and two $z$-edges  emanating from vertex $\textbf{1}$, as in the  fourth diagram  in  Figure \ref{fig:C=42}. Then there are  $z_{23}$-stroke, $z_{24}$-stroke  and $z_{34}$-stroke by Rule VI.  Note that vertex $\textbf{1}$ is $w$-circled, then the circling method implies all vertices except possibly vertex $\textbf{5}$   are $w$-circled. Then there is $w$-stroke emanating from $\textbf{3}$ and vertex $\textbf{4}$.  The $w$-stroke may go to vertex $\textbf{5}$,  or it is a  $w_{34}$-stroke.

If  all two  $w$-strokes go to vertex $\textbf{5}$, then Rule VI implies the existence of $w_{34}$-stroke, which contradicts with $C=4$.  Then the $w$-strokes from $\textbf{3}$ and vertex $\textbf{4}$ is $w_{34}$-stroke, and  vertex $\textbf{5}$   is disconnected.

Consider the circling the the diagram. We have three different cases: all vertices are not $z$-circled, only two of the first four vertices, say, vertex $\textbf{1}$ and vertex $\textbf{2}$ are $z$-circled, or all the first four vertices are $z$-circled.

If none of the first four vertices is $z$-circled, then  we have $\Gamma_{1234}=0$ and $L_{1234}=0$ by Rule IV and Proposition \ref{Prp:LI}.  This is a contradiction.

If only vertex $\textbf{1}$ and vertex $\textbf{2}$ are $z$-circled,  we have $\Gamma_1+\Gamma_2=0$ by Rule IV. On the other hand, Corollary \ref{Cor:sumT12}  implies that $\Gamma_1+\Gamma_2\ne0$.  This is a contradiction.

Then all the first four vertices are $z$-circled. This gives the   fourth diagram  in  Figure \ref{fig:C=42}. However, this contradicts with Proposition \ref{Prp:dumbbell},  so excluded. 

Hence,  there is no  possible  diagram, i.e.,  the fourth one  in Figure \ref{fig:C=42} is impossible.

\subsubsection{}
Suppose that from vertex $\textbf{1}$ there are three  $w$-edges go to vertex $\textbf{2}$,   vertex $\textbf{3}$  and vertex $\textbf{4}$  respectively and one $z$-edges goes to vertex $\textbf{5}$  .  There is a $w$-stroke between any pair of $\{ \textbf{1, 2, 3, 4} \}$ by Rule VI, and the four vertices are all $z$-circled.

Then there are $z$-strokes emanating from vertex $\textbf{2}$,   vertex $\textbf{3}$  and vertex $\textbf{4}$ .
None of them can go to vertex $\textbf{5}$   by Rule VI. Then there are $z_{23}$-, $z_{24}$- and $z_{34}$-strokes. This contradict with $C=4$. Hence, there is no possible diagram in this case. 

\begin{figure}[!h]
	
	\centering
	\includegraphics[width=0.8\textwidth]{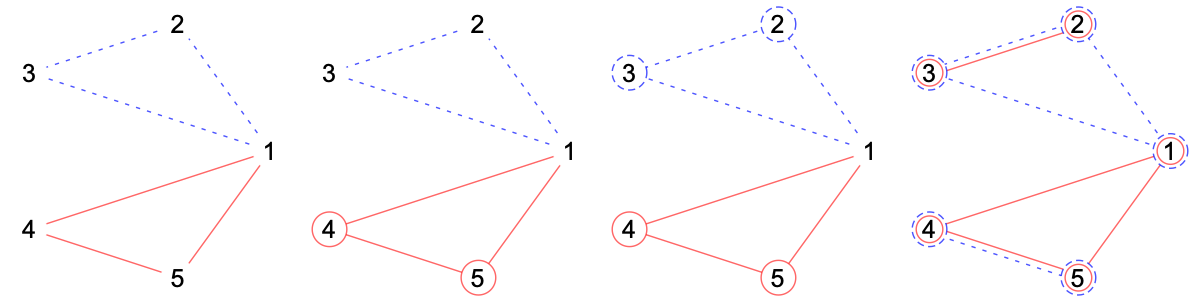}

	\caption{Four diagrams for $C=4$, no $zw$-edges. They correspond to Diagram  3 of Figure \ref{fig:list2},  Diagram  3 of Figure \ref{fig:list1}, Diagram  17 of Figure \ref{fig:list2}, and Diagram  9 of Figure \ref{fig:list1} respectively.    }
	\label{fig:C=43}
\end{figure}

\subsubsection{}
Suppose that there are two  $w$-edges and two $z$-edges  emanating from vertex $\textbf{1}$, with numeration as in the first diagram of Figure \ref{fig:C=43}.
By Rule VI, there is $w_{23}$- and $z_{45}$-stroke. Thus, we have two attached triangles. By Rule VI, if there is more stroke, it must be $z_{23}$- and/or $w_{45}$-stroke. \\

\emph{Case I}: If there is no more stroke,  then vertex $\textbf{1}$ can not be circled, otherwise, the circling method would lead to contradiction.
Rule IV implies that vertex $\textbf{2}$ and vertex $\textbf{3}$  can only be both or both not  $w$-circled, and
vertex $\textbf{4}$  and vertex $\textbf{5}$   can only be  both or both not  $z$-circled.   Then we have the  first three  diagrams in Figure \ref{fig:C=43}. 

\emph{Case II}:  If there is  only $z_{23}$-stroke, the circling method implies vertex $\textbf{1}$, vertex $\textbf{2}$ and vertex $\textbf{3}$  are $z$-circled. Note that only vertex $\textbf{2}$ and vertex $\textbf{3}$  can be $w$-circled. There are two possibilities: whether vertex $\textbf{2}$ and vertex $\textbf{3}$  are both  $w$-circled or both  not $w$-circled.

If vertex $\textbf{2}$ and vertex $\textbf{3}$  are both  $w$-circled, we  have $\Gamma_2+\Gamma_3=0$ by Rule IV. On the other hand, since  vertices $\textbf{1}$,  $\textbf{4}$, and $\textbf{5}$ are not $w$-circled, they are not $w$-close to vertex $\textbf{2}$ and vertex $\textbf{3}$ . Then
Proposition \ref{Prp:sumT12} implies that $\Gamma_2+\Gamma_3\ne0$.  This is a contradiction.

If vertex $\textbf{2}$ and vertex $\textbf{3}$  are both not  $w$-circled, we  have $\Gamma_2+\Gamma_3=0$  and $L_{123}=0$ by Rule IV and Proposition \ref{Prp:LI}. This is a contradiction. Hence there is no possible diagram in this case.

\emph{Case III}: If there are  $z_{23}$- and $w_{45}$-strokes, the circling method implies vertex $\textbf{2}$ and vertex $\textbf{3}$  are $z$-circled, vertex $\textbf{4}$  and vertex $\textbf{5}$   are $w$-circled, and vertex $\textbf{1}$ is  $z$- and $w$-circled. 
By Rule I,  vertex $\textbf{2}$ and vertex $\textbf{3}$ are both $z$-circled, and vertex $\textbf{4}$ and vertex $\textbf{5}$ are both $w$-circled.  This is  the last diagram  in Figure \ref{fig:C=43}. 

Hence, there are  \emph{four} possible diagrams in this case, as shown in Figure \ref{fig:C=43}. 





 In summary,  among the thirteen diagrams in Figure \ref{fig:C=41}, \ref{fig:C=42} and \ref{fig:C=43}, we have excluded the first and third one in Figure \ref{fig:C=41} and the fourth one in Figure \ref{fig:C=42}.   We have \emph{10} possible  diagrams.

\subsection{$C=5$}

There are three  cases: two $zw$-edges,  one $zw$-edge with one $z$-edge and two  $w$-edges,    one $zw$-edge with three $z$-edges.

\subsubsection{}
Suppose that there are two $zw$-edges and one $z$-edges  emanating from vertex $\textbf{1}$, with numeration as in the first digram of Figure \ref{fig:C=51}.  Rule VI implies the  existence of $z_{23}$-,  $w_{23}$-, $z_{24}$- and $z_{34}$-strokes. Then vertex $\textbf{5}$   can be either disconnected  or connects with vertex $\textbf{4}$  by one  $w$-edge, otherwise, it would contradict with $C=5$.

For the disconnected diagram, note that any of the connected four vertices can not be $w$-circled. Otherwise, the circling method implies vertex $\textbf{4}$  is $w$-circled, which is a contradiction.  There are three cases: none of the vertices are circled,  all four vertices except vertex $\textbf{4}$  are $z$-circled, or all four vertices are $z$-circled.

If none of the vertices are circled,  by Proposition \ref{Prp:LI} we have   $L_{123}=0$ since  vertices $\textbf{1}$,  $\textbf{2}$, and $\textbf{3}$ form a triangle with no $w$-circled attached. Similarly, we have $L_{1234}=0$ by Proposition \ref{Prp:LI}. This is a contradiction since
\begin{equation*}
L_{1234}=L_{123}+\Gamma_4(\Gamma_1+\Gamma_2+\Gamma_3), \  (\Gamma_1+\Gamma_2+\Gamma_3)^2-2L_{123}\ne 0.
\end{equation*}

If only  vertices $\textbf{1}$,  $\textbf{2}$, and $\textbf{3}$ are $z$-circled, we also have  $L_{123}=0$. By Rule IV, we have $\Gamma_1+\Gamma_2+\Gamma_3=0$. This is a contradiction.

Then we have only one possible diagram for the disconnected diagram, and it is the first in Figure \ref{fig:C=51}.

For the connected diagram with $w_{45}$-edge, the circling method implies that all vertices are $w$-circled. Note that vertex $\textbf{4}$  and vertex $\textbf{5}$   can not be $z$-circled. Then we have two cases: all the five vertices are not $z$-circled, or  vertices $\textbf{1}$,  $\textbf{2}$, and $\textbf{3}$ are $z$-circled.

If all the five vertices are not $z$-circled, we have $\Gamma_1+\Gamma_2+\Gamma_3=0$ by Rule IV and $L_{1234}=0$ by Proposition \ref{Prp:LI}. This is a contradiction. 

Then, we have only one possible diagram for the connected diagram, and it is the second in Figure \ref{fig:C=51}.

\begin{figure}[!h]
	
	\centering
	\includegraphics[width=0.4\textwidth]{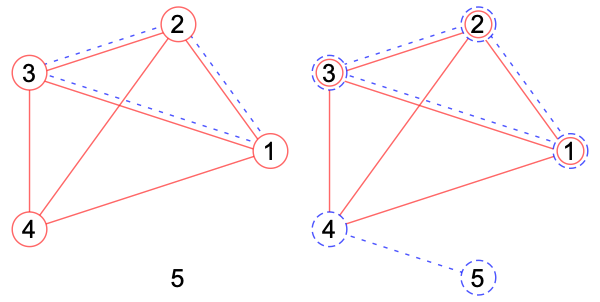}

	\caption{Two diagrams for $C=5$, two $zw$-edges. They correspond to Diagram  5 of Figure \ref{fig:list1},  and Diagram  7 of Figure \ref{fig:list1} respectively.    }
	\label{fig:C=51}
\end{figure}

Hence, we have \emph{two} possible diagrams, shown in Figure \ref{fig:C=51}.

\subsubsection{}

Suppose that there are one $zw$-edge, one $z$-edge, and two  $w$-edges emanating from vertex $\textbf{1}$, with numeration as in the first diagram of Figure \ref{fig:C=52}.  Rule VI implies that  vertices $\textbf{1}$,  $\textbf{2}$, $\textbf{3}$, and $\textbf{4}$ are fully $w$-stroked, and that
vertices $\textbf{1}$,  $\textbf{3}$, and $\textbf{5}$ are fully $z$-stroked. There is no more stroke emanating from  vertices $\textbf{1}$,  $\textbf{3}$, and $\textbf{5}$, since that would contradict with $C=5$. There are two cases, $z_{24}$-stroke is present or not.

If it is not present,  vertex $\textbf{2}$ and vertex $\textbf{4}$  can only be $w$-circled, and vertex $\textbf{5}$   can only be $z$-circled. Then vertex $\textbf{1}$ and vertex $\textbf{3}$   can not be circled, otherwise, the circling method would lead to contradiction.  Then we have the first two diagrams in Figure \ref{fig:C=52} by Rule IV.

If $z_{24}$-stroke is present, then  vertices $\textbf{1}$,  $\textbf{2}$, $\textbf{3}$, and $\textbf{4}$ are all $z$-circled. Vertex $\textbf{1}$ and vertex $\textbf{3}$  can not be $w$-circled, which would lead to vertex $\textbf{5}$   also $w$-circled and a contradiction. Vertex $\textbf{2}$ and vertex $\textbf{4}$  are $w$-close, so they are both $w$-circled or both not $w$-circled. If they are both $w$-circled,  then Rule IV implies that
$\Gamma_2+\Gamma_4=0$. On the other hand, Proposition \ref{Prp:sumT12} implies that
$\Gamma_2+\Gamma_4\ne0$. This is a contradiction.

Then according to whether vertex $\textbf{5}$   is $z$-circled or not, we have two cases, which are the last two diagrams in
Figure  \ref{fig:C=52}.

The third diagram in Figure  \ref{fig:C=52} is impossible.  We would have $L_{1234}=0$ and $\sum_{j=1}^4\Gamma_j=0$ by Proposition \ref{Prp:LI} and Rule IV.  Hence, we have only one  possible diagram if  $z_{24}$-stroke is present, and it is the last  diagram in Figure \ref{fig:C=52}.

\begin{figure}[h!]
	
	\centering
	\includegraphics[width=0.8\textwidth]{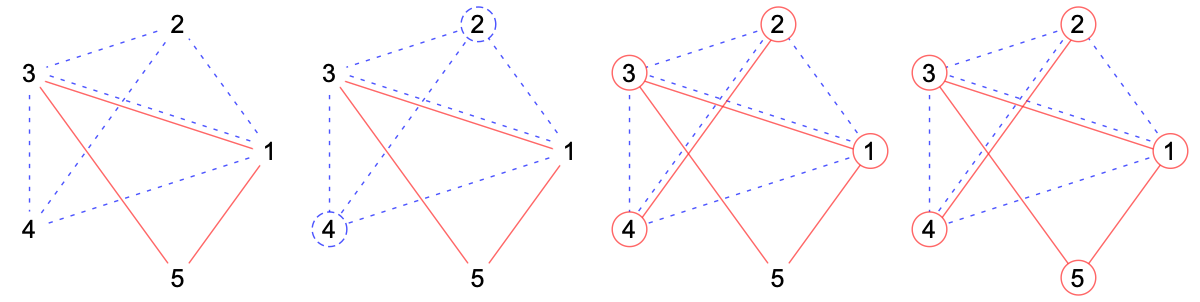}
	
	\caption{Four diagrams for $C=5$, one $zw$-edges. The  third one has been excluded.The first, second and the fourth one  correspond to Diagram  4 of Figure \ref{fig:list2},  Diagram  4 of Figure \ref{fig:list1},  and Diagram  12 of Figure \ref{fig:list2} respectively.    }
	\label{fig:C=52}
\end{figure}

Hence, we have \emph{three} possible diagrams, the first two and the last one in  Figure \ref{fig:C=52}. 

\subsubsection{}

Suppose that there is one $zw$-edge and three  $w$-edges emanating from vertex $\textbf{1}$ .  Let us say, the $zw$-edge goes to vertex $\textbf{2}$ and the other edges go to the other three vertices.  Rule VI implies that  vertices $\textbf{1, 2, 3, 4, 5}$ are fully $w$-stroked. The circling method implies that all vertices are $z$-circled. Then there are $z$-strokes emanating from vertices $\textbf{3}$, $\textbf{4}$, and $\textbf{5}$. Since $C=5$ at vertex $\textbf{1}$ and vertex $\textbf{2}$,  the $z$-strokes from vertices $\textbf{3}$, $\textbf{4}$, and $\textbf{5}$ must go to vertices $\textbf{3}$, $\textbf{4}$, and $\textbf{5}$. By Rule VI, they form a triangle of $z$-strokes. This is a contradiction with $C=5$.  Hence, there is no possible diagram in this case.

In summary,  among the six  diagrams in Figure \ref{fig:C=51} and \ref{fig:C=52}, we have excluded the third one in Figure \ref{fig:C=52}.   We have \emph{five} possible  diagrams.

\subsection{$C=6$}

There are three  cases: three  $zw$-edges,  two $zw$-edge with one edge in each color,   two $zw$-edge with two $z$-edge.

\begin{figure}[!h]

	\centering
	\includegraphics[width=0.4\textwidth]{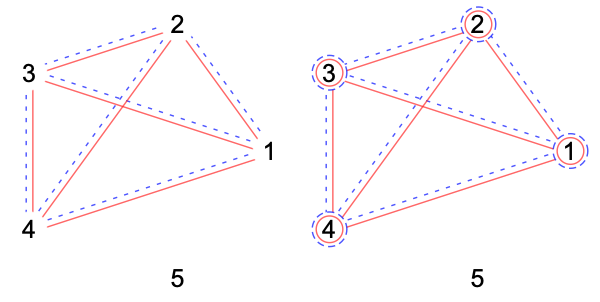}

	\caption{Two diagrams for $C=6$, three $zw$-edges. The  second one has been excluded. The first one corresponds to Diagram 10 of Figure \ref{fig:list2}. }
	\label{fig:C=61}
\end{figure}

\subsubsection{}
Suppose that there are three $zw$-edges  emanating from vertex $\textbf{1}$, with numeration as in the first digram of Figure \ref{fig:C=61}.  Rule VI implies that the  vertices $\textbf{1}$,  $\textbf{2}$, $\textbf{3}$, and $\textbf{4}$ are fully $zw$-edged. Then vertex $\textbf{5}$   must be disconnected  since $C=6$. The first four vertices can be circled in the same way by the circling method.

If all the first four vertices are $z$-circled but not $w$-circled, then we would have
$L_{1234}=0$ and $\sum_{j=1}^4\Gamma_j=0$ by Proposition \ref{Prp:LI} and Rule IV.  This is one contradiction.

Then we have two possible diagrams, as shown in Figure \ref{fig:C=61}. However, the second one is impossible by Proposition \ref{Prp:quadrilateral}.

Hence, there is only \emph{one}  possible diagram if   there are three $zw$-edges, the first one   in Figure \ref{fig:C=61}.

\subsubsection{}
Suppose that there are two $zw$-edges and two $z$-edges emanating from vertex $\textbf{1}$, with numeration as in the first diagram of Figure \ref{fig:C=62}.  Rule VI implies the existence of $z_{25}$-, $z_{35}$-,  $z_{24}$-, $z_{45}$- and $z_{34}$-strokes,   and  that the  vertices $\textbf{1}$,  $\textbf{2}$, and $\textbf{3}$ are fully $zw$-edged. If there is more stroke, it must be $w_{45}$-stroke since $C=6$.

If there is no more stroke, then the vertices can only be $z$-circled. There are two cases, either  vertices $\textbf{1}$,  $\textbf{2}$, and $\textbf{3}$ are $z$-circled or not.

If  vertices $\textbf{1}$,  $\textbf{2}$, and $\textbf{3}$ are $z$-circled, then either vertex $\textbf{4}$  or vertex $\textbf{5}$   must also be $z$-circled. Otherwise, we have $\Gamma_1+\Gamma_2+\Gamma_3=0$ and $L_{123}=0$ by  Proposition \ref{Prp:LI} and Rule IV.  This is one contradiction. Then we have the first two diagrams in Figure \ref{fig:C=62}.

If  vertices $\textbf{1}$,  $\textbf{2}$, and $\textbf{3}$ are not  $z$-circled, then vertex $\textbf{4}$  and vertex $\textbf{5}$   are both $z$-circled or both not $z$-circled. Then we have the second two  diagrams in Figure \ref{fig:C=62}.

If $w_{45}$-stroke is present, then the circling method implies that all vertices are $w$-circled. Rule IV implies that $\Gamma_{123}=0$. If at least one, hence all three vertices $\textbf{1}$,  $\textbf{2}$, and $\textbf{3}$ are  $z$-circled. 
Then there are two cases, according to whether vertices $\textbf{4}$ and $\textbf{5}$ are  $z$-circled.  Hence,  we have the last  two  diagrams in Figure \ref{fig:C=62}.  If none of vertices $\textbf{1}$,  $\textbf{2}$, and $\textbf{3}$ are  $z$-circled,  there are also  two cases, depending on whether vertices $\textbf{4}$ and $\textbf{5}$ are  $z$-circled.  If  the two vertices are both $z$-circled, Rule IV implies that $\Gamma_{45}=0$, which contradicts with  $\Gamma_{45}\ne 0$ by Proposition \ref{Prp:sumT12}.  If  the two vertices are both not $z$-circled, Proposition \ref{Prp:LI} yields $L_{12345}=0$. But Rule IV implies that $\Gamma_{12345}=0$, a contradiction.

\begin{figure}[!h]

	\centering
	\includegraphics[width=0.6\textwidth]{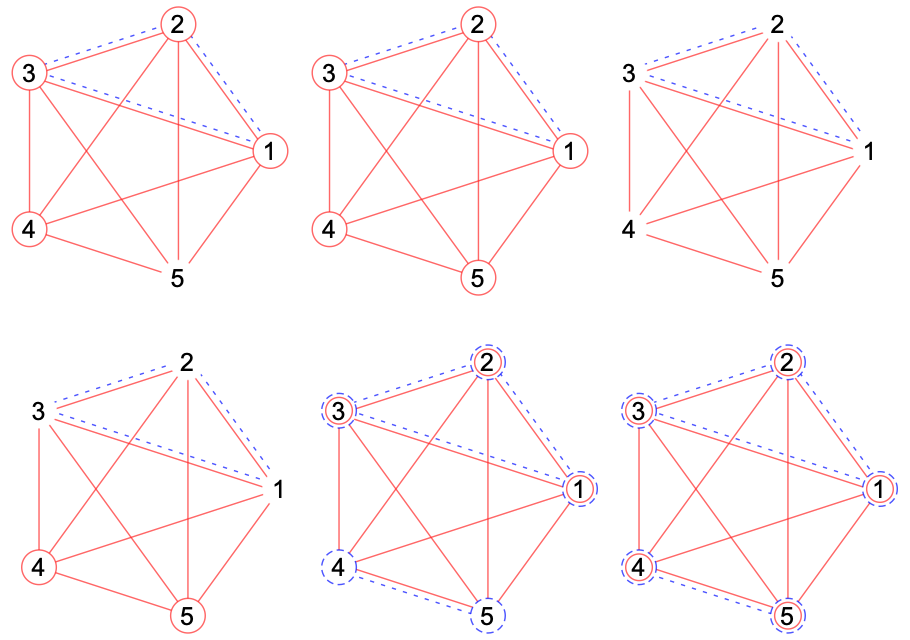}

	\caption{Six diagrams for $C=6$, two  $zw$-edges, case 1. They correspond to Diagram  7 of Figure \ref{fig:list2},  Diagram  9 of Figure \ref{fig:list2}, Diagram  5 of Figure \ref{fig:list2},  Diagram  6 of Figure \ref{fig:list2},  Diagram  8 of Figure \ref{fig:list1}, and Diagram  21 of Figure \ref{fig:list2} respectively.    } 
	\label{fig:C=62}
\end{figure}

Hence, there are \emph{six} possible diagrams in this case, as shown in Figure \ref{fig:C=62}.

\subsubsection{}
Suppose that there are two $zw$-edges and one edge of each color  emanating from vertex $\textbf{1}$, with numeration as in the first diagram of Figure \ref{fig:C=63}.  Rule VI implies the existence of $w_{25}$-, $w_{35}$-,  $z_{24}$-,  and $z_{34}$-strokes,   and  that the  vertices $\textbf{1}$,  $\textbf{2}$, and $\textbf{3}$ are fully $zw$-edged. If there is more stroke, it must be $w_{45}$-stroke, but it would lead to extra stroke emanating from vertex $\textbf{1}$, which contradicts with $C=6$.

Vertex $\textbf{5}$  can only be $w$-circled, and vertex $\textbf{4}$  can only be $z$-circled. Then  vertices $\textbf{1}$,  2 and  3  can not be circled, otherwise, the circling method would lead to contradiction.  Then we have the  diagram in Figure \ref{fig:C=63} by Rule IV. 

\begin{figure}[h!]

	\centering
	\includegraphics[width=0.2\textwidth]{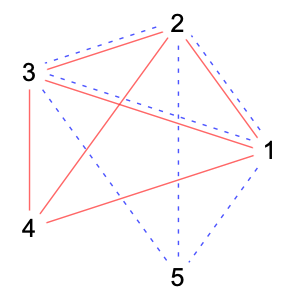}

	\caption{One diagram for $C=6$, two  $zw$-edges, case 2. It corresponds to Diagram 11 of Figure \ref{fig:list2}. }
	\label{fig:C=63}
\end{figure}

In summary,   among the nine diagrams of  Figure \ref{fig:C=61}, \ref{fig:C=62},   and \ref{fig:C=63}, we exclude the second one of 
	Figure \ref{fig:C=61}. That is, we have \emph{eight} possible diagrams.

\subsection{$C=7$}
Suppose that there are three $zw$-edges and one $z$-edge  emanating from vertex $\textbf{1}$, with numeration as in the  diagram of Figure \ref{fig:C=7}.  Rule VI implies that the  vertices $\textbf{1}$,  $\textbf{2}$, $\textbf{3}$, and $\textbf{4}$ are fully $zw$-edged and that  vertex $\textbf{5}$   connects with the first four vertices by one $z$-edge. There is no more stroke since $C=7$. If any vertex is $w$-circled, all are $w$-circled, which would lead to $w$-stroke emanating from vertex $\textbf{5}$   . This is a contradiction.  There are two cases, either vertex $\textbf{1}$ is $z$-circled or not.

If vertex $\textbf{1}$ is not $z$-circled, then none of the vertices is $z$-colored by the circling method. Then $L_{1234}=0$ and $L=0$ by  Proposition \ref{Prp:LI}, a contradiction.

If vertex $\textbf{1}$ is $z$-circled,  then  vertices $\textbf{1}$,  $\textbf{2}$, $\textbf{3}$, and $\textbf{4}$ are all $z$-circled. In this case,  vertex $\textbf{5}$   must be  $z$-circled. Otherwise, 
we have  $L_{1234}=0$ and $\sum_{j=1}^4\Gamma_j=0$ by  Proposition \ref{Prp:LI} and Rule IV,  a contradiction.

Hence, we only have \emph{one} diagram in the case of $C=7$, as in Figure \ref{fig:C=7}.

\begin{figure}[!h]

	\centering
	\includegraphics[width=0.2\textwidth]{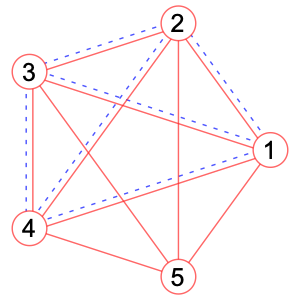}

	\caption{One diagram for $C=7$. It corresponds to Diagram 13 of Figure \ref{fig:list2}.}
	\label{fig:C=7}
\end{figure}


\subsection{$C=8$}
Suppose that there are four $zw$-edges   emanating from vertex $\textbf{1}$ . Rule VI implies that the  vertices $\textbf{1, 2, 3, 4, 5}$ are fully $zw$-edged. Since all vertices are both $z$-close and $w$-close, the vertices are    all $z$-circled (respectively, $w$-circled) or all not $z$-circled (respectively, $w$-circled). 
If they are all just $z$-circled but not $w$-circled, then we have 
\[  \Gamma=0, \ \  L=0,  \]
by Rule IV and Proposition \ref{Prp:LI}, a contradiction.

\begin{figure}[!h]

	\centering
	\includegraphics[width=0.4\textwidth]{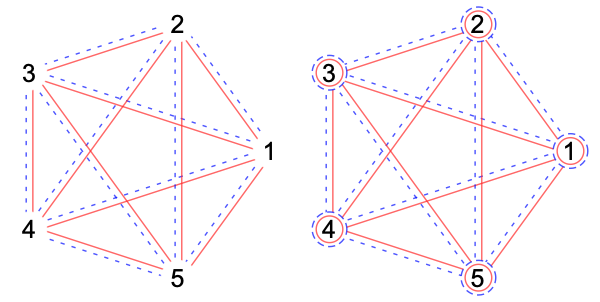}

	\caption{Two diagrams for $C=8$. They correspond to Diagram 14 and 22 of Figure \ref{fig:list2}. }
	\label{fig:C=8}
\end{figure}

Hence, we  have \emph{two}  diagrams in the case of $C=8$, as in Figure \ref{fig:C=8}. \\

\emph{Summary}: In the searching for all problematic 5-vortex  two-colored diagrams,  we have found 39 of them, as shown in Figure \ref{fig:C=0}-\ref{fig:C=8}. Among them, we have  excluded eight of them, i.e., the first diagram in Figure \ref{fig:C=2}, the two diagrams in Figure \ref{fig:C=3}, the second and fourth diagram in Figure \ref{fig:C=41}, the fourth diagram  in Figure \ref{fig:C=42}, the third diagram in Figure \ref{fig:C=52}, and
the second diagram in Figure \ref{fig:C=61}. Hence, we conclude that  any singular sequence should converge to one of the remaining  31 diagrams.

\section{Further diagram exclusion and the vorticity constraints }
\label{sec:diagram&constraints}

\indent\par 
   The 31 diagrams found in Section \ref{sec:matrixrules} are collected into two classes, \emph{9 diagrams}
\emph{ that can be excluded} by  further arguments, see Figure \ref{fig:list1} , and 
\emph{ 22 diagrams that we can not exclude},  see Figure  \ref{fig:list2}.

	We would like to point out that the diagrams in the two lists differ in appearance from those in Section \ref{sec:matrixrules}. In Figure \ref{fig:list1} and \ref{fig:list2}, the diagrams are ordered by the number of circles, whereas in Section \ref{sec:matrixrules}, they are ordered by the maximal number of strokes from a two-colored vertex. Additionally, there are differences in vertex labeling and in the switching of $z$- and $w$-diagrams. These differences are not mistakes. Each diagram in Section \ref{sec:matrixrules} represents an equivalence class under vertex permutations and color switching.  It is therefore valid to present different representatives of the same equivalence classes, and we have intentionally done so to emphasize this fact, as it is often overlooked by readers.

For the  22 diagrams in Figure \ref{fig:list2}, 
we list the vorticity constraints.  Most of them are straightforward from the results
in Section \ref{sec:moreproperty},  while some  requires  additional work.  
	We will use notations such as \(\Gamma_J, \Gamma_{j_1, \ldots, j_n}, L_J,\) and \(L_{j_1, \ldots, j_n}\) below. Please refer to Definition \ref{def:LI} for their meanings.

\subsection{The 9 impossible diagrams}

The 9 impossible diagrams are presented in Figure \ref{fig:list1}. We exclude them in the following. 

\begin{figure}[!h]

	\centering
	\includegraphics[width=\textwidth]{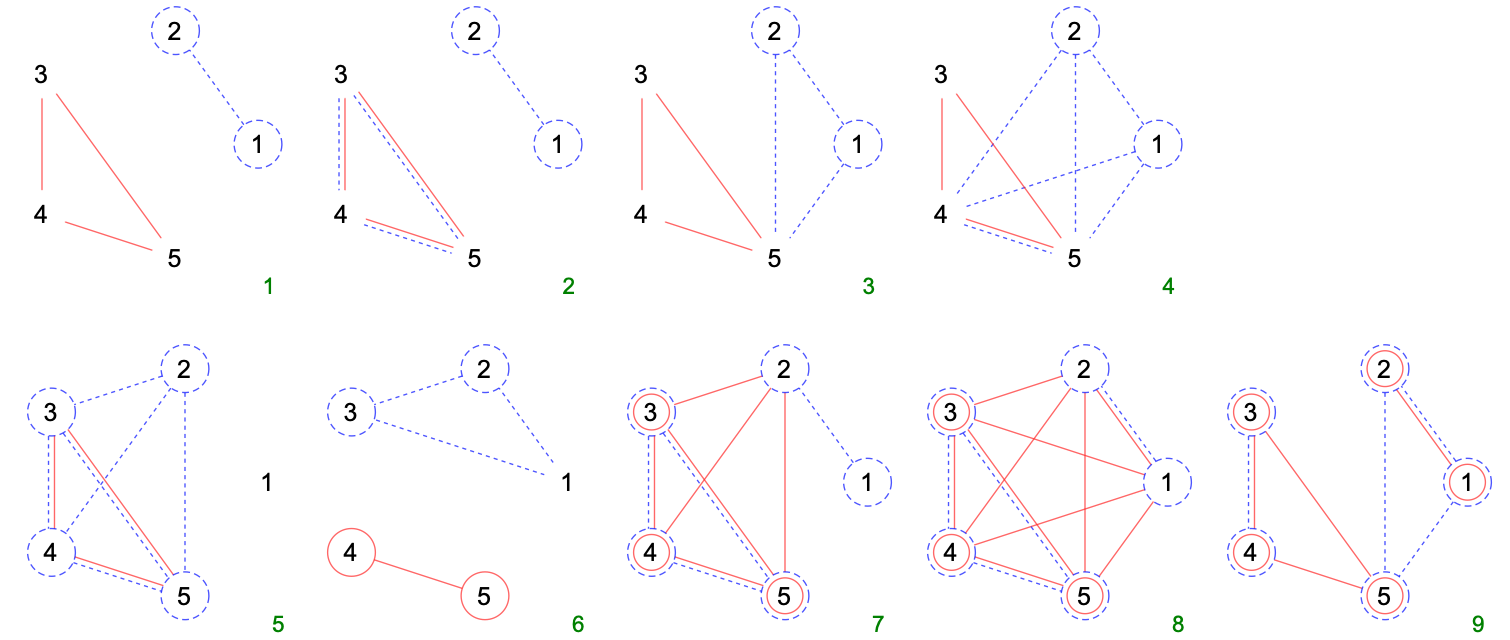}

	\caption{The  9 diagrams that can be excluded. They correspond to the second diagram of Figure \ref{fig:C=0}, the  third diagram of Figure \ref{fig:C=41},   the  second diagram of Figure \ref{fig:C=43},   the  second diagram of Figure \ref{fig:C=52},  the  first diagram of Figure \ref{fig:C=51},  the  third diagram of Figure \ref{fig:C=0},   the  second diagram of Figure \ref{fig:C=51},  the  fifth diagram of Figure \ref{fig:C=62}, and the  fourth diagram of Figure \ref{fig:C=43} respectively.       }
	\label{fig:list1}
\end{figure}

\begin{enumerate}
	
	\item  Diagram 1:  Consider the isolated component
	$\{3,4,5\}$of the $z$-digram, we see $L_{345}=0$; Consider the
	isolated component $\{1,2\}$  of the $w$-diagram. By Proposition \ref{Prp:isolate-z_12}, there are two cases:
	\begin{enumerate}
		\item $\Lambda=\pm1$, then $\Gamma_{3}+\Gamma_{4}+\Gamma_{5}=0$, which
		contradicts with $L_{345}=0;$
		\item $\Lambda=\pm \mathbf{i}$, then $L=0,$and $\Gamma_{1}\Gamma_{2}=L_{345}$.
		Since $L_{345}=0,$ we obtain $\Gamma_{1}\Gamma_{2}=0,$ which is
		impossible. 
	\end{enumerate}
	
	\item Diagram 2 can be excluded by the same argument as  that of
	Diagram 1.  
	
	\item  Diagram 3:  Consider the isolated component
	$\{3,4,5\}$of the $z$-diagram, we see $L_{345}=0$; Consider the
	isolated component $\{1,2, 3\}$of the $w$-diagram. By Proposition \ref{Prp:isolate-z_123}, there are two cases:
	\begin{enumerate}
		\item $\Lambda=\pm1$, then $\Gamma_{3}+\Gamma_{4}=0$, which contradicts
		with $L_{345}=0$ since $L_{345}=\Gamma_{3}\Gamma_{4}+\Gamma_{5}(\Gamma_{3}+\Gamma_{4})$.
		\item $\Lambda=\pm \mathbf{i}$, then $L=0,$and $L_{125}=L_{345}$. Since $L_{345}=0,$
		we obtain $\Gamma_{3}\Gamma_{4}=-\Gamma_{5}(\Gamma_{3}+\Gamma_{4})$;
		similarly, it also holds $\Gamma_{1}\Gamma_{2}=-\Gamma_{5}(\Gamma_{1}+\Gamma_{2})$.
		Note that 
		\[
		0=L=L_{125}+(\Gamma_{3}+\Gamma_{4})(\Gamma_{1}+\Gamma_{2}+\Gamma_{5})+\Gamma_{3}\Gamma_{4}=(\Gamma_{3}+\Gamma_{4})(\Gamma_{1}+\Gamma_{2}),
		\]
		which implies that $\Gamma_{1}\Gamma_{2}\Gamma_{3}\Gamma_{4}=0,$
		which
		is impossible. 
	\end{enumerate}
	
	\item Diagram 4:  	We will show that $z_1\succ \epsilon^2$, then it contradicts with Corollary \ref{Cor:sumT12} and this diagram is thus excluded.  
	
	Note that $w_{1j}, w_{2j}\approx  \epsilon^{-2}, j\ge3$. Then 
	\[\Lambda\sum_{j\ge 3} \Gamma_j z_j=\sum_{j\ge 3}\frac{\Gamma_1\Gamma_j}{w_{1j}}+\sum_{j\ge3}\frac{\Gamma_2\Gamma_j}{w_{2j}}\preceq \epsilon^{2}.\]
	Then the equation $\sum_{j=3}^{5}\Gamma_j z_j=\Gamma_{345} z_3+(\Gamma_{4}+\Gamma_{5})z_{34}+\Gamma_{5}z_{45}$ and the fact that $ z_{34}\succ \epsilon^{2},  z_{45}\approx \epsilon^{2}$ implies 
	\[  \Gamma_{345} z_3+(\Gamma_{4}+\Gamma_{5})z_{34} \preceq \epsilon^2. \]
	By Proposition \ref{Prp:LI}, it follows that $L_{345}=0$ and then $\Gamma_{345}\ne 0,~\Gamma_{4}+\Gamma_{5}\ne 0$. Then it holds that $z_3\approx z_{34}\succ \epsilon^2.$
	
	Finally, the equation $\sum_{j}\Gamma_j z_j=0$ and the fact that $ z_{jk}\approx \epsilon^2,j\ne k, j,k\in\{1,2,4,5\}$ lead to 
	\[  \Gamma_3 z_3 + \Gamma_{1245}z_1\preceq \epsilon^2. \]
	Hence, it holds that $z_1 \approx z_3 \succ \epsilon^2$. 
	
	\item Diagram 5: By Proposition \ref{Prp:LI}, it follows that $L_{345}=0$ and $\Gamma_{345}\ne 0,~\Gamma_{3}+\Gamma_{4}\ne 0$. 
	
	We first claim that \emph{the vertices \textbf{2} and \textbf{3} are not $w$-close, or,  $w_{23}\approx \epsilon^{-2}$}, whose proof will be given later.  Since $\frac{1}{w_{1j}} \approx \epsilon^{2},  j\ge2$, we have
	$\Lambda z_1=\sum_{j=2}^{5}\frac{\Gamma_j}{w_{1j}}\preceq \epsilon^{2}. $
	Note that $z_{12}\succ \epsilon^{2}$, so $z_2 \succ \epsilon^2$. Then it is easy to see
	$ z_2\sim z_3\sim z_4\sim z_5\succ \epsilon^{2} \succeq z_1.$
	By the claim, $\frac{1}{w_{2j}} \approx  \epsilon^{2},  j\ge3$, so 
	\[\Lambda\sum_{j=3}^{5}\Gamma_j z_j=\sum_{j=3}^{5}\frac{\Gamma_1\Gamma_j}{w_{1j}}+\sum_{j=3}^{5}\frac{\Gamma_2\Gamma_j}{w_{2j}}\preceq \epsilon^{2}. \]
	By the equation  $\sum_{j=3}^{5}\Gamma_j z_j=\sum_{j=3}^{5}\Gamma_j z_3 +\Gamma_4z_{34} +\Gamma_5 z_{35}$, we conclude that 
	$\sum_{j=3}^{5}\Gamma_j=0$, 
	which is a contradiction.

	We now prove the claim  by contradiction. 
	Suppose that   \textbf{2} and \textbf{3} are $w$-close. Rule IV implies that   $\sum_{j=2}^{5}\Gamma_j=0.$ Note that $\frac{1}{w_{1j}}- \frac{1}{w_{12}} \prec  \epsilon^{2},  j\ge3$, so
	\[\Lambda\sum_{j=2}^{5}\Gamma_j z_j=\sum_{j=2}^{5}\frac{\Gamma_1\Gamma_j}{w_{1j}}=\Gamma_1 \sum_{j=2}^{5} \Gamma_j  (\frac{1}{w_{12}}+\frac{1}{w_{1j}}- \frac{1}{w_{12}} )   \prec \epsilon^{2}.\]
	Then  \[  \epsilon^2 \succ \sum_{j=2}^{5}\Gamma_j z_j=\sum_{j=3, 4, 5}\Gamma_j z_{2j}\Rightarrow \Gamma_3 z_{23}\sim -\Gamma_4 z_{24}-\Gamma_5 z_{25}\approx \epsilon^{2}.  \]
	Note that the scenario is now  the same as that described in Proposition \ref{Prp:quadrilateral}, thus impossible.

	\item  Diagram 6: By Proposition  \ref{Prp:isolate-z_123}, it holds
	that $\Gamma_{4}+\Gamma_{5}\ne0$. There are two cases:
	\begin{enumerate}
		\item $\Lambda=\pm1$. Consider the isolated component $\{2,3\}$
		of the $w$-diagram. Then  $\Gamma_{4}+\Gamma_{5}=0$, which is a contradiction. 
		\item $\Lambda=\pm \mathbf{i}$. Then $L=0.$ Consider the isolated component
		$\{4,5\}$of the $z$-diagram. Then $\Gamma_{4}\Gamma_{5}=L_{123}$. 
		On the other hand,  consider the isolated component $\{2,3\}$ of the $w$-diagram. Then 
		$L_{145}=L_{123}$. Thus,  it holds 
		\[
		\Gamma_{4}\Gamma_{5}=L_{145}=\Gamma_{4}\Gamma_{5}+\Gamma_{1}(\Gamma_{4}+\Gamma_{5}),
		\]
		which is a contradiction. 
	\end{enumerate}

	\item Diagram 7:  It is easy to see that  $w_2\sim w_3\sim w_4\sim w_5\approx  \epsilon^{-2}, \Gamma_1 w_1 +\Gamma_2 w_2\prec  \epsilon^{-2}$, and 
	\[\Gamma_{3}+\Gamma_{4}+\Gamma_{5}=0. \]
	Set $ w_2 \sim b \epsilon^{-2}$. Then $\ w_1\sim -\frac{\Gamma_2}{\Gamma_1}b \epsilon^{-2}$. 
	We claim that $\Gamma_1+\Gamma_2 \ne 0$. Otherwise, we have 
	\[\Gamma_2 w_{12}=\Gamma_1 w_1+\Gamma_2 w_2=\Lambda  (\sum_{j=3}^{5}\frac{\Gamma_j\Gamma_1}{z_{j1}}+ \sum_{j=3}^{5}\frac{\Gamma_j\Gamma_2}{z_{j2}})\prec \epsilon^{2},\] which contradicts with Proposition \ref{Estimate1}. Hence, $\Gamma_1+\Gamma_2 \ne 0$ and $w_{12}\approx \epsilon^{-2}$.

	By $\Lambda z_{1}=\sum_{j=3}^{5}\frac{\Gamma_j}{w_{j1}}+\frac{\Gamma_2}{w_{21}}$ and $\sum_{j=3}^{5}\frac{\Gamma_j}{w_{j1}}\prec \epsilon^2$,
	it follows that 
	\[ z_{1}\sim -\frac{\Gamma_1\Gamma_2 }{b \Lambda(\Gamma_1+\Gamma_2 )}\epsilon^{2}.\]
	Similarly, by $w_{1}=\Lambda\sum_{j=3}^{5}\frac{\Gamma_j}{z_{j1}}+\frac{\Gamma_2}{z_{21}}$ and $\sum_{j=3}^{5}\frac{\Gamma_j}{z_{j1}}\prec \epsilon^2$,
	it follows that 
	\[ z_{12}\sim \frac{\Gamma_1\Lambda }{b}\epsilon^{2}.\]
	Since $\sum_{j=1}^{5} \Gamma_j z_j=0$, we obtain 
	\[\sum_{j=3}^{5}\Gamma_j z_j=-(\Gamma_1 +\Gamma_2)z_1-\Gamma_2 z_{12}=\left\{
	\begin{array}{lr}
	\sim (\frac{1}{\Lambda }-\Lambda )\frac{\Gamma_1\Gamma_2}{b}\epsilon^{2} &  if~\Lambda^2\ne 1;\\
	\prec \epsilon^{2},  & if ~\Lambda^2=1 .
	\end{array}
	\right.\]
	Note that $\sum_{j=3}^{5}\Gamma_j z_j=\Gamma_{4}z_{34}+\Gamma_{5}z_{35}=\Gamma_{3}z_{43}+\Gamma_{5}z_{45}$. Set $\frac{z_{34}}{\Gamma_{5}}\sim c\epsilon^{2}.$
	
	\emph{Case 1.} $\Lambda^2=1$: Then
	$\frac{z_{34}}{\Gamma_{5}}\sim\frac{z_{45}}{\Gamma_{3}}\sim\frac{z_{53}}{\Gamma_{4}}\sim c\epsilon^{2}.$
	Note that $ \bar{\Lambda}  w_{3}\sim  \frac{\Gamma_4}{z_{43}}+ \frac{\Gamma_5}{z_{53}}$ and $ \bar{\Lambda} w_{4}\sim  \frac{\Gamma_3}{z_{34}}+ \frac{\Gamma_5}{z_{54}}$. Then
	\[  b c \Gamma_4\Gamma_5+\Gamma_4^2 \Lambda -\Gamma_5^2 \Lambda =0, b c \Gamma_3\Gamma_5+\Gamma_3^2 (-\Lambda )+\Gamma_5^2 \Lambda=0.  \] 
	Eliminating $bc$, we obtain
	$\frac{(\Gamma_3+\Gamma_4) \left(\Gamma_3\Gamma_4-\Gamma_5^2\right)}{\Gamma_4} \Lambda =0,$
	which is  a contradiction since $\Gamma_3+\Gamma_4 =-\Gamma_5\ne 0$ and $\Gamma_3\Gamma_4-\Gamma_5^2=L_{345}\ne0$.

	\emph{Case 2. } $\Lambda^2\ne 1$: Then
	\[\frac{z_{45}}{\Gamma_{3}}\sim [c+ (\frac{1}{\Lambda }-\Lambda )\frac{\Gamma_1\Gamma_2}{b\Gamma_3\Gamma_5}]\epsilon^{2}\approx \epsilon^{2},\ \frac{z_{35}}{\Gamma_{4}}\sim [-c+ (\frac{1}{\Lambda }-\Lambda )\frac{\Gamma_1\Gamma_2}{b\Gamma_4\Gamma_5}]\epsilon^{2}\approx \epsilon^{2}.\]
	Note that $ \bar{\Lambda}  w_{3}\sim  \frac{\Gamma_4}{z_{43}}+ \frac{\Gamma_5}{z_{53}}$ and $ \bar{\Lambda} w_{4}\sim  \frac{\Gamma_3}{z_{34}}+ \frac{\Gamma_5}{z_{54}}$. Then
	\begin{align*}
		&b^2 c^2 \Gamma_4\Gamma_5^2 \Lambda +b c \Gamma_1\Gamma_2\Gamma_5 \Lambda ^2-b c \Gamma_1\Gamma_2\Gamma_5+b c \Gamma_4^2 \Gamma_5 \Lambda ^2-b c \Gamma_5^3 \Lambda ^2+\Gamma_1\Gamma_2\Gamma_4\Lambda ^3-\Gamma_1\Gamma_2\Gamma_4\Lambda =0, \\
		&b^2 c^2 \Gamma_3\Gamma_5^2 \Lambda -b c \Gamma_1\Gamma_2\Gamma_5 \Lambda ^2+b c \Gamma_1\Gamma_2\Gamma_5-b c \Gamma_3^2 \Gamma_5 \Lambda ^2+b c \Gamma_5^3 \Lambda ^2+\Gamma_1\Gamma_2\Gamma_3\Lambda ^3-\Gamma_1\Gamma_2\Gamma_3\Lambda=0.  
	\end{align*}
	Eliminating $bc$, we obtain $\Lambda^2 =  \frac{\Gamma_1\Gamma_2}{\Gamma_1\Gamma_2+\Gamma_3\Gamma_4-\Gamma_5^2}$. Hence,  
	$\Lambda^2 \in \mathbb{R}$.  It must be $-1$, and then $\Lambda=\pm \sqrt{-1}, L=0$, which is a contradiction since
	\[  L=\Gamma_1\Gamma_2+(\Gamma_1+\Gamma_2)(\Gamma_3 +\Gamma_4+\Gamma_5 )+L_{345}= \Gamma_1\Gamma_2+\Gamma_3\Gamma_4-\Gamma_5^2.\]

	\item Diagram 8:   It is easy to see that $\Gamma_{3}+\Gamma_{4}+\Gamma_{5}=0$, $\Gamma_{1}+\Gamma_{2}=0$. 
	Note that $\frac{1}{z_{31}}- \frac{1}{z_{j1}}, \frac{1}{z_{32}}- \frac{1}{z_{j2}} \prec  \epsilon^{2},  j\ge3$, so
	\[\bar{\Lambda}(\Gamma_1 w_1 +\Gamma_2 w_2)=
	\Gamma_1 \sum_{j=3}^{5} \Gamma_j  (\frac{1}{z_{31}}+\frac{1}{z_{j1}}- \frac{1}{z_{31}} )  +\Gamma_2 \sum_{j=3}^{5} \Gamma_j  (\frac{1}{z_{32}}+\frac{1}{z_{j2}}- \frac{1}{z_{32}} )   \prec \epsilon^{2}.\]
	Hence
	\[  \epsilon^{2} \succ \Gamma_1 w_1 +\Gamma_2 w_2=( \Gamma_1 +\Gamma_2) w_1 +\Gamma_2 w_{12}= \Gamma_2 w_{12},  \]
	which contradicts with Proposition \ref{Estimate1}.

	\item Diagram 9:  By Considering the $z$- and $w$-diagrams, we obtain 
	\[\Gamma_{1}+\Gamma_{2}=0,\Gamma_{3}+\Gamma_{4}=0,\]
	and 
	\[ \Gamma_{1}z_1+\Gamma_{2}z_2=  \Gamma_{2}z_{12}\approx \epsilon^2,   \Gamma_{3}z_3+\Gamma_{4}z_4=  \Gamma_{4}z_{34}\approx \epsilon^2.  \]
	By $\sum_{j=1}^{5}\Gamma_j z_j=0$, we see $\Gamma_5 z_5\preceq \epsilon^{2}$,   which is a contradiction.
	
\end{enumerate}

\subsection{The remaining 22 diagrams}\label{subsec:diagrams}

\begin{figure}[!h]

	\centering
	\includegraphics[width=\textwidth]{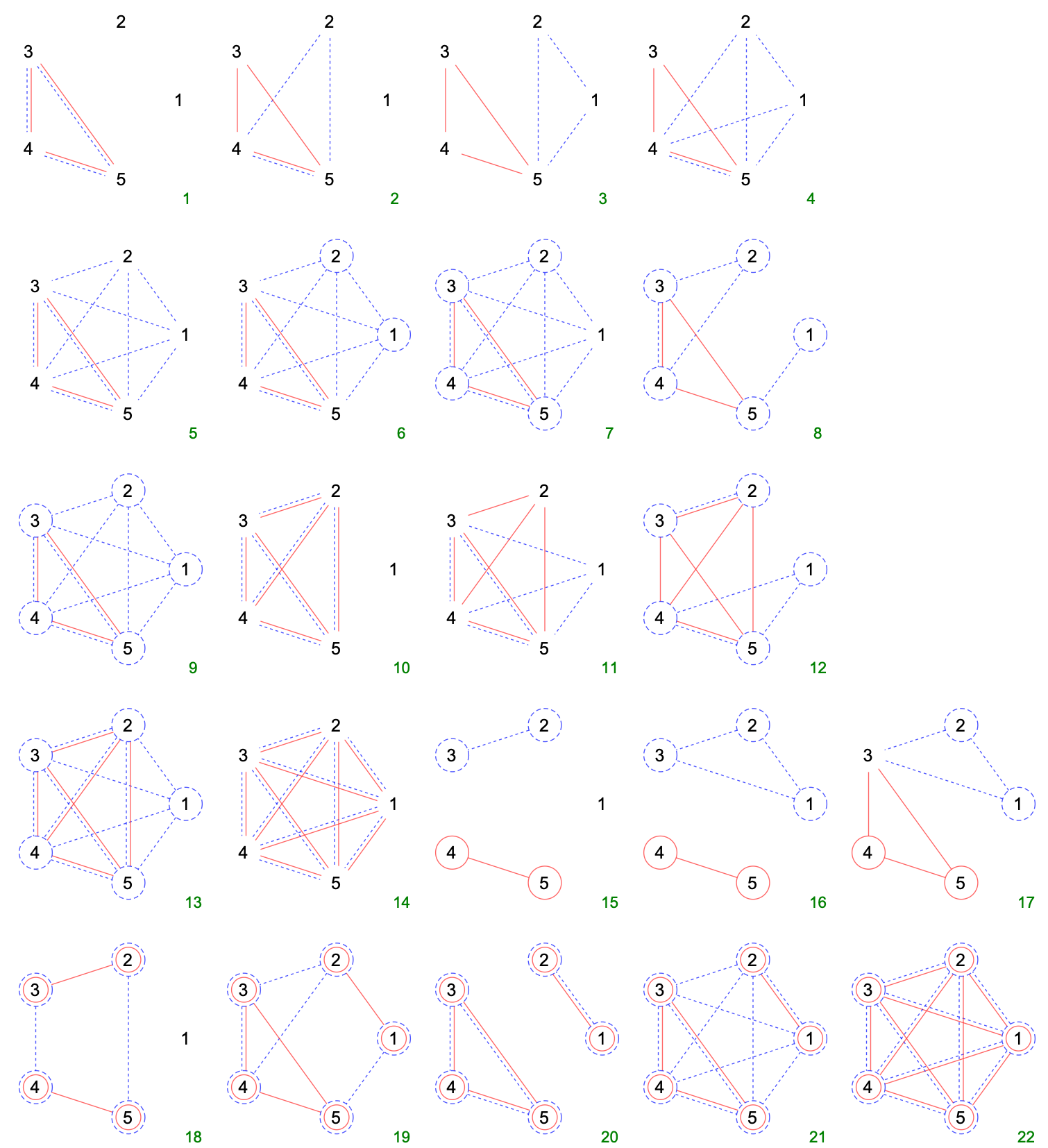}

	\caption{ The remaining 22 diagrams. They correspond to the first diagram of Figure \ref{fig:C=41}, the first diagram of Figure \ref{fig:C=42}, the first diagram of Figure \ref{fig:C=43}, the first diagram of Figure \ref{fig:C=52}, 
		the  third diagram of Figure \ref{fig:C=62}, 	the  fourth diagram of Figure \ref{fig:C=62}, 	the  first diagram of Figure \ref{fig:C=62}, 	the  second diagram of Figure \ref{fig:C=42},
			the  second diagram of Figure \ref{fig:C=62}, 	the  first diagram of Figure \ref{fig:C=61}, the diagram of Figure \ref{fig:C=63}, 	the  fourth diagram of Figure \ref{fig:C=52},
			the diagram of Figure \ref{fig:C=7}, 	the  first diagram of Figure \ref{fig:C=8}, the  first diagram of Figure \ref{fig:C=0},  the  fourth diagram of Figure \ref{fig:C=0}, 
			the  third diagram of Figure \ref{fig:C=43},   the  second diagram of Figure \ref{fig:C=2},
			 the  third diagram of Figure \ref{fig:C=42},   the  fifth diagram of Figure \ref{fig:C=41},  
			  the  sixth diagram of Figure \ref{fig:C=62},   and the  second diagram of Figure \ref{fig:C=8} respectively.  }
	\label{fig:list2}
\end{figure}

\begin{enumerate}
	\item Diagram 1: $L_{345}=0.$ \emph{This is the case of Roberts' example.}
	
	\item Diagram 2: $L_{345}=L_{245}=0$, or equivalently, $L_{345}=0,$$\Gamma_{2}=\Gamma_{3}.$
	
	\item Diagram 3: $\Lambda=\pm1$, $L_{345}=L_{125}=0$. 
	
	It is easy to see that $L_{345}=L_{125}=0$. We claim that $\Lambda=\pm1$. Otherwise, we have $L=0$. 
	Then 
	\begin{align*}
		L&=L_{125} + \Gamma_3 \Gamma_4 + (\Gamma_1 +\Gamma_2 +\Gamma_5)(\Gamma_3 +\Gamma_4)\\
		&= L_{125} + L_{345}+(\Gamma_1 +\Gamma_2)(\Gamma_3 +\Gamma_4)=(\Gamma_1 +\Gamma_2)(\Gamma_3 +\Gamma_4)=0.
	\end{align*}
	Without loss of generality, assume that $\Gamma_1 +\Gamma_2=0$. Then $L_{125}=0$ implies that 
	$\Gamma_1 \Gamma_2=0$,	which is  a contradiction. Hence,  it holds that $\Lambda=\pm1$. 
	
	\item Diagram 4: $L_{345}=L_{1245}=0$. 
	
	\item Diagram 5: $L_{345}=L=0$. 
	
	\item Diagram 6: $L_{345}=0.$

	\item Diagram 7: $L_{345}=0.$ 
	
	\item Diagram 8: 	$L_{345}=0, [\Gamma_{1}(\Gamma_{3}+\Gamma_{4})-\Gamma_{2}\Gamma_5](\Gamma_{1}+\Gamma_5)(\Gamma_{2}+\Gamma_{3}+\Gamma_{4})=0$. 
	
	It is easy to see $L_{345}=0$. 
	Set  $w_{3}\sim w_{4}\sim w_{5}\sim a\epsilon^{-2}$. Then $w_1\sim -\frac{\Gamma_5}{\Gamma_1}a \epsilon^{-2}$, and $w_2 \sim -\frac{\Gamma_3+\Gamma_4}{\Gamma_2}a \epsilon^{-2}$.  If $w_{12}\prec \epsilon^{-2}$,  then  
	\[   \Gamma_{1}(\Gamma_{3}+\Gamma_{4})=\Gamma_{2}\Gamma_5.   \]
	We now assume that $w_{12}\approx  \epsilon^{-2}$. 
	Since $z_{12}\succ    \epsilon^2$, it holds that  $z_1\succ \epsilon^2$, or  $z_2 \succ \epsilon^2$. 
	
	\emph{Case I: $z_1\succ \epsilon^2$. } Then
	$\epsilon^2\prec \Lambda  z_1 =\sum_{j=3}^{5}\frac{\Gamma_j}{w_{j1}}+\frac{\Gamma_2}{w_{21}}$, and  
	$\sum_{j=3}^{5}\frac{\Gamma_j}{w_{j1}}
	\succ \epsilon^2.$ If  $\Gamma_{1}+\Gamma_5\ne 0$, then $\frac{1}{w_{31}}, \frac{1}{w_{41}}, \frac{1}{w_{51}}\approx \epsilon^{2}
	$, which is a contradiction. Thus, we obtain 
	\[ \Gamma_{1}+\Gamma_5= 0. \]
	
	\emph{Case II: $z_2\succ \epsilon^2$. } Then
	$\epsilon^2\prec \Lambda  z_2 =\sum_{j=3}^{5}\frac{\Gamma_j}{w_{j2}}+\frac{\Gamma_1}{w_{12}}$, and  
	$\sum_{j=3}^{5}\frac{\Gamma_j}{w_{j2}}
	\succ \epsilon^2.$ If  $\Gamma_{2}+\Gamma_{3}+\Gamma_{4}\ne 0$, then $\frac{1}{w_{32}}, \frac{1}{w_{42}}, \frac{1}{w_{52}}\approx \epsilon^{2}
	$, which is a contradiction. Thus, we obtain 
	\[ \Gamma_{2}+\Gamma_{3}+\Gamma_{4}= 0. \]

	\item Diagram 9: $L_{345}=0.$ 
	
	\item Diagram 10: $L_{2345}=0.$
	
	\item Diagram 11: $L_{1345}=L_{2345}=0$, or equivalently, $L_{1345}=0,$$\Gamma_{1}=\Gamma_{2}.$
	
	\item Diagram 12: $L_{2345}=0=\Gamma_{2}+\Gamma_{3}.$
	
	\item Diagram 13: $L_{2345}=0.$
	
	\item Diagram 14: $L=0$.	It is easy to see that $r_{kl}\approx \epsilon^2$ for any $k\ne l, k, l\in \{1, 2, 3, 4, 5\}$. 
	
	\item Diagram 15. $\Lambda=\pm1$, $\Gamma_{1}+\Gamma_{2}+\Gamma_{3}=0$, $\Gamma_{1}+\Gamma_{4}+\Gamma_{5}=0$.
	
	By Proposition \ref{Prp:isolate-z_12}, there are two cases:
	\begin{enumerate}
		\item $\Lambda=\pm1$, then $\Gamma_{1}+\Gamma_{2}+\Gamma_{3}=0$, and $\Gamma_{1}+\Gamma_{4}+\Gamma_{5}=0$.
		\item $\Lambda=\pm \mathbf{i}$, then $L=0,$  and $\Gamma_{2}\Gamma_{3}=L_{145},\Gamma_{4}\Gamma_{5}=L_{123}$.  Then it holds 
		\begin{align*}
			L&=L_{145} + \Gamma_2 \Gamma_3 + (\Gamma_1 +\Gamma_4 +\Gamma_5)(\Gamma_2 +\Gamma_3)\\
			&= L_{145} + L_{123}+(\Gamma_4 +\Gamma_5)(\Gamma_2 +\Gamma_3)\\
			&=\Gamma_{2}\Gamma_{3}+\Gamma_{4}\Gamma_{5}+(\Gamma_4 +\Gamma_5)(\Gamma_2 +\Gamma_3)=L_{2345}=0.
		\end{align*}
		On the other hand, $L=\Gamma_1\Gamma_{2345}+L_{2345}=0$ implies that $\Gamma_{2345}=0$, 	which is  a contradiction. Hence, this case  is impossible. 
	\end{enumerate}

	\item Diagram 16: Consider the isolated component $\{4,5\}$ of the $z$-digram.
	By Proposition \ref{Prp:isolate-z_12}, it holds that $\Gamma_{4}+\Gamma_{5}\ne0$.
	There are two cases:
	\begin{enumerate}
		\item $\Lambda=\pm1$, then $\Gamma_{1}+\Gamma_{2}+\Gamma_{3}=0$. 
		\item $\Lambda=\pm \mathbf{i}$, then $L=0$ and $\Gamma_{4}\Gamma_{5}=L_{123}$.
	\end{enumerate}
	It is easy to see that
	\[   r_{14}, r_{15}, r_{24}, r_{25}, r_{34}, r_{35}\approx \epsilon^{-2}. \]
	
	According to Proposition \ref{Prp:isolate-z_12}, $z_{45}$ is a maximal stroke.  	\emph{We claim that $w_{12}, w_{13}, w_{23}$ are all maximal strokes}.  Therefore, by Proposition \ref{Estimate1}, 
	\[  r_{12}, r_{13}, r_{23}, r_{45} \approx 1.  \]

	Now, we prove the above claim, i.e.,   $w_1\nsim w_2 \nsim w_3 \nsim w_1$.  
	Since,  $w_{41}, w_{51}\approx \epsilon^{-2}$, by Proposition \ref{Prp:triangle2}, it can not happen that $w_1\sim w_2\sim w_3$.  If exactly two of $w_1, w_2, w_3$ are $w$-close, without loss of generality, assume that $w_{12}\prec \epsilon^{-2}, w_{13}\approx \epsilon^{-2}$. Then 
	\begin{align*}
		\Lambda z_1=\frac{\Gamma_2}{w_{21}} + \sum_{ j=3}^5 \frac{\Gamma_j}{w_{j1}}\succ \epsilon^{2}, \  \Lambda z_3= \sum_{ j=1, j\ne 3}^5 \frac{\Gamma_j}{w_{j3}}\preceq \epsilon^{2}.
	\end{align*}
	Hence, $z_{13}\sim -z_1 \succ \epsilon^{2}$, which  contradicts with Proposition \ref{Estimate1}. Hence, the claim is proved.

	\item Diagram 17: $\Lambda=\pm \mathbf{i}$, $L=0$ and $L_{345}=L_{123}$. 
	
	Consider
	the isolated component $\{3,4,5\}$ of the $z$-diagram and the component
	$\{1,2,3\}$ of the $w$-diagram. By Proposition \ref{Prp:isolate-z_123}, it holds that
	$\Gamma_{4}+\Gamma_{5}\ne0$ and $\Gamma_{1}+\Gamma_{2}\ne0$. There
	are two cases:
	\begin{enumerate}
		\item $\Lambda=\pm1$, then $\Gamma_{1}+\Gamma_{2}=0$, which is a contradiction. Hence, this case  is impossible. 
		\item $\Lambda=\pm \mathbf{i}$, $L=0$ and $L_{345}=L_{123}$.
	\end{enumerate}
	It is easy to see that
	\[   r_{14}, r_{15}, r_{24}, r_{25}\approx \epsilon^{-2}. \]
	According to Proposition \ref{Prp:isolate-z_123}, $z_{45}, w_{12}$ are maximal strokes. It is easy to see that $z_{34}, z_{35}, w_{23}, w_{13}$ are all maximal stokes. By Proposition \ref{Estimate1}, it holds that 
	\[  r_{12}, r_{13}, r_{23}, r_{34}, r_{35}, r_{45}\approx 1.   \]

	\item Diagram 18: $\Gamma_{3}\Gamma_{5}=\Gamma_{2}\Gamma_{4}$. \emph{This
		is the case of Roberts' example. }
	
	Consider the isolated component
	$z_{23}$-edge. By Rule III, $\Gamma_{2}z_{2}+\Gamma_{3}z_{3}\prec\epsilon^{-2}$.
	Without loss of generality, assume $z_{2}\sim-\Gamma_{3}a\epsilon^{-2}$
	and $z_{3}\sim\Gamma_{2}a\epsilon^{-2}$.Similarly, we assume $z_{4}\sim-\Gamma_{5}c\epsilon^{-2}$
	and $z_{5}\sim\Gamma_{4}c\epsilon^{-2}$, where $a$ and $c$ are
	two non-zero complex constants. Since $z_{34}\approx z_{25}\approx\epsilon^{2}$
	by the estimate of Proposition \ref{Estimate1}, we obtain $\Gamma_{2}a=-\Gamma_{5}b,-\Gamma_{3}a=\Gamma_{4}b.$
	Hence, we obtain 
	\[
	\Gamma_{2}\Gamma_{4}=\Gamma_{3}\Gamma_{5}.
	\]
	Apply similar argument to the isolated component $w_{25}$-edge. Then
	we have 
	\[
	\begin{array}{cc}
	z_{2}\sim z_{5}\sim-\Gamma_{3}a\epsilon^{-2}, & z_{3}\sim z_{4}\sim\Gamma_{2}a\epsilon^{-2},\\
	w_{2}\sim w_{3}\sim-\Gamma_{5}b\epsilon^{-2}, & w_{4}\sim w_{5}\sim\Gamma_{2}b\epsilon^{-2},
	\end{array}
	\]
	where $b$ is a non-zero complex constants. It follows that 
	\[
	\Gamma_{2}\Gamma_{4}=\Gamma_{3}\Gamma_{5}.
	\]
	If $\Gamma_{2}+\Gamma_{3}+\Gamma_{4}+\Gamma_{5}\neq0$, then $(\Gamma_{5}+\Gamma_{2})(\Gamma_{2}+\Gamma_{3})(\Gamma_{3}+\Gamma_{4})(\Gamma_{4}+\Gamma_{5})\neq0$.
	Otherwise, without loss of generality, assume $\Gamma_{2}+\Gamma_{3}=0$,
	it follows that $\Gamma_{4}+\Gamma_{5}=0$, which contradicts with  $\Gamma_{2}+\Gamma_{3}+\Gamma_{4}+\Gamma_{5}\neq0$.
	
	In this case,  it is easy to see that
	\begin{align*}
		r_{23}^{2}\sim\frac{\Gamma_{2}+\Gamma_{3}}{-\Lambda},r_{34}^{2}\sim\frac{\Gamma_{3}+\Gamma_{4}}{-\overline{\Lambda}}, & r_{45}^{2}\sim\frac{\Gamma_{4}+\Gamma_{5}}{-\Lambda},r_{25}^{2}\sim\frac{\Gamma_{2}+\Gamma_{5}}{-\overline{\Lambda}};\\
		r_{24}^{2}\sim(\Gamma_{2}+\Gamma_{3})(\Gamma_{2}+\Gamma_{5})ab\epsilon^{-4}, & r_{35}^{2}\sim-(\Gamma_{2}+\Gamma_{3})(\Gamma_{2}+\Gamma_{5})ab\epsilon^{-4}.;\\
		r_{12}^{2}\sim\Gamma_{3}\Gamma_{5}ab\epsilon^{-4},r_{13}^{2}\sim-\Gamma_{2}\Gamma_{5}ab\epsilon^{-4}, & r_{14}^{2}\sim\Gamma_{2}\Gamma_{2}ab\epsilon^{-4},r_{15}^{2}\sim-\Gamma_{2}\Gamma_{3}ab\epsilon^{-4}.
	\end{align*}

	\item Diagram 19: $\Gamma_{1}(\Gamma_{3}+\Gamma_{4})=\Gamma_{2}\Gamma_{5}$. Moreover, it holds that either $|\Gamma_{1}+\Gamma_{2}|+|\Gamma_{1}+\Gamma_{5}|=0$ or  $(\Gamma_{1}+\Gamma_{2})(\Gamma_{1}+\Gamma_{5} )\ne 0$.

	Consider the isolated component $z_{12}$-edge. By Rule III, $\Gamma_{1}z_{1}+\Gamma_{2}z_{2}\prec\epsilon^{-2}$.
	Without loss of generality, assume $z_{1}\sim-\Gamma_{2}a\epsilon^{-2}$
	and $z_{2}\sim\Gamma_{1}a\epsilon^{-2}$. Similarly, we assume $w_{1}\sim-\Gamma_{5}b\epsilon^{-2}$
	and $w_{5}\sim\Gamma_{1}b\epsilon^{-2}$, where $a$ and $b$ are
	two non-zero complex constants. Then 
	\[
	\begin{array}{cc}
	z_{1}\sim z_{5}\sim-\Gamma_{2}a\epsilon^{-2}, & z_{2}\sim z_{3}\sim z_{4}\sim\Gamma_{1}a\epsilon^{-2},\\
	w_{1}\sim w_{2}\sim-\Gamma_{5}b\epsilon^{-2}, & w_{3}\sim w_{4}\sim w_{5}\sim \Gamma_{1}b\epsilon^{-2}.
	\end{array}
	\]
	Since $\sum_{ j=1}^5 \Gamma_j z_j =0$, we obtain $\Gamma_{1}(\Gamma_{3}+\Gamma_{4})=\Gamma_{2}\Gamma_{5}$. 
	Note that if $\Gamma_{1}+\Gamma_{2}=0,$ then $\Gamma_{3}+\Gamma_{4}+\Gamma_{5}=0;$
	if $\Gamma_{1}+\Gamma_{5}=0,$ then $\Gamma_{2}+\Gamma_{3}+\Gamma_{4}=0.$
	\begin{enumerate}
		\item If it holds that $\Gamma_{1}+\Gamma_{2}=0$ and $\Gamma_{1}+\Gamma_{5}=0,$
		i.e., $\Gamma_{1}=\Gamma_{3}+\Gamma_{4}=-\Gamma_{2}=-\Gamma_{5}.$
		It holds that $z_{j}\sim\Gamma_{1}a\epsilon^{-2},w_{j}\sim\Gamma_{1}b\epsilon^{-2},z_{jk}\prec\epsilon^{-2},w_{jk}\prec\epsilon^{-2}$ for any $j,k\in \{1, 2, 3, 4, 5\}$. 
		Thus, 
		\[
		\begin{array}{cc}
		\epsilon^{4}\prec r_{12},r_{15},r_{23},r_{24},r_{35},r_{45}\prec1, & r_{34}\approx\epsilon^{2},\\
		\epsilon^{4}\prec r_{14}^{2}=z_{14}w_{14}\approx z_{45}w_{45}\prec\epsilon^{-4}, & \epsilon^{4}\prec r_{13},r_{25}\prec\epsilon^{-4}.
		\end{array}
		\]
		\item If it holds that $\Gamma_{1}+\Gamma_{2}\ne0$ and $\Gamma_{1}+\Gamma_{5}\ne0.$
		It holds that $z_{1}\sim z_{5}\sim-\Gamma_{2}a\epsilon^{-2},z_{2}\sim z_{3}\sim z_{4}\sim\Gamma_{1}a\epsilon^{-2},w_{1}\sim w_{2}\sim-\Gamma_{5}b\epsilon^{-2},w_{3}\sim w_{4}\sim w_{5}\sim\Gamma_{1}b\epsilon^{-2}$.
		Thus, 
		\begin{align*}
			r_{34}^{2}\approx\epsilon^{4}, & r_{12},r_{35},r_{45}\approx1,\\
			r_{15},r_{23},r_{24}\approx1, & r_{14}^{2},r_{13}^{2},r_{25}^{2}\approx\epsilon^{-4}.
		\end{align*}
		
		\item If it holds that $\Gamma_{1}+\Gamma_{2}=0$ and $\Gamma_{1}+\Gamma_{5}\ne0.$
		It holds that $z_{j}\sim\Gamma_{1}a\epsilon^{-2},z_{jk}\prec\epsilon^{-2}$  for any $j,k \in \{1, 2, 3, 4, 5\}$, $w_{1}\sim w_{2}\sim-\Gamma_{5}b\epsilon^{-2},w_{3}\sim w_{4}\sim w_{5}\sim\Gamma_{1}b\epsilon^{-2}$.  
		Note that 
		\[\Lambda\sum_{j=3}^{5}\Gamma_j z_j=\sum_{j=3}^{5}\frac{\Gamma_1\Gamma_j}{w_{1j}}+\sum_{j=3}^{5}\frac{\Gamma_2\Gamma_j}{w_{2j}}=\sum_{j=3}^{5}\frac{\Gamma_2\Gamma_j w_{12}}{w_{1j}w_{2j}}\prec \epsilon^{2}\]
		and $ \sum_{j=3}^{5}\Gamma_j z_j=\sum_{j=3}^{5}\Gamma_j z_{3}+\Gamma_{4}z_{34}+\Gamma_{5}z_{35},\ z_{34}\approx \epsilon^{2}$.  It holds that 
		\[z_{35}\approx \epsilon^{2},\]		
		which is a contradiction. Hence, this case is impossible. 	
		
	\end{enumerate}

	\item Diagram 20: $\Gamma_{1}+\Gamma_{2}=0,\Gamma_{3}+\Gamma_{4}+\Gamma_{5}=0$.

	\item Diagram 21: $\Gamma_{1}+\Gamma_{2}=0,\Gamma_{3}+\Gamma_{4}+\Gamma_{5}=0$.
	
	\item Diagram 22: $\Gamma_{1}+\Gamma_{2}+\Gamma_{3}+\Gamma_{4}+\Gamma_{5}=0$. 	It is easy to see that $r_{kl}\approx \epsilon^2$ for any $k\ne l, k, l\in \{1, 2, 3, 4, 5\}$. 
\end{enumerate}

	\section{Proofs of the main results} \label{sec:proof}

\indent\par
\begin{proof}[\emph{Proof of Theorem \ref{thm:Main}}]
	Suppose that there are infinitely many solutions of system \eqref{equ:complexcc} in the complex domain. At least one of the squared distances \(r^2_{kl}\), say \(r_{12}^2\), must take infinitely many values. By Lemma \ref{Eliminationtheory}, \(r_{12}^2 = z_{12}w_{12}\) is dominating. There exists a sequence of complex normalized  central configurations such that \(r_{12}^2 \to 0\). Then \(z_{12}^{(n)} w_{12}^{(n)} \to 0\), and either \(\mathcal{Z}\) or \(\mathcal{W}\) is unbounded along this sequence. We extract a singular sequence, which must correspond to one of the 22 diagrams in Figure \ref{fig:list2}. Consequently, some explicit linear or quadratic relations on the five vorticities must be satisfied.

	Consider first the case of  relative equilibria. We can further exclude Diagram 17. Recall that \(r_{12}^2\) is dominating. By pushing it to \(\infty\), we cannot be in Diagram 14 or 22. Thus, there are finitely many relative equilibria unless at least one of the following fourteen polynomial systems holds. ( Recall that  $L_{j_1, ..., j_n}=\sum_{1 \leq k < l \leq n} \Gamma_{j_k} \Gamma_{j_l}$ ), 
	\begin{equation}
	\quad \begin{aligned}
	&           L_{345}=0;\\
	&      	  L_{345}=0, \Gamma_2=\Gamma_3; \\
	&	  	      L_{345}=L_{125}=0;\\
	&  	  	  L_{345}=L_{1245}=0;\\
	&  	  	  L_{345}=L=0;\\
	&          L_{345}=0, [\Gamma_{1}(\Gamma_{3}+\Gamma_{4})-\Gamma_{2}\Gamma_5](\Gamma_{1}+\Gamma_5)(\Gamma_{2}+\Gamma_{3}+\Gamma_{4})=0;\\
	&  	  	  L_{2345}=0;\\
	&  	  	  L_{1345}=0, \Gamma_{1}-\Gamma_{2}=0;\\
	&  	  	  L_{2345}=0, \Gamma_{2}+\Gamma_{3}=0;\\
	&  	  	  \Gamma_{1}+\Gamma_{2}+\Gamma_{3}=0, \Gamma_{1}+\Gamma_{4}+\Gamma_{5}=0; \\
	& 	  	  \Gamma_{1}+\Gamma_{2}+\Gamma_{3}=0;\\
	&  	  	  \Gamma_3\Gamma_5=\Gamma_2\Gamma_4; \\
	& 	  	  \Gamma_1(\Gamma_3+\Gamma_4)=\Gamma_2\Gamma_5; \\
	&  	  	  \Gamma_{1}+\Gamma_{2}=0,  \Gamma_{3}+\Gamma_{4}+\Gamma_{5}=0. 
	\end{aligned} \label{equ:constraint_relative_equilibria}
	\end{equation}

	Now consider the case of collapse configurations. In this case, we have an additional constraint \(L = 0\). It is straightforward to see that we can further exclude 10 diagrams, namely, Diagram 3, 4, 10, 11, 12, 13, 15, 20, 21, and 22. Hence, for collapse configurations, we have only 12 remaining diagrams: Diagram 1, 2, 5, 6, 7, 8, 9, 14, 16, 17, 18, and 19. Recall that \(r_{12}^2\) is dominating. By pushing it to \(\infty\), we cannot be in Diagram 14. Therefore, there are finitely many collapse configurations unless at least one of the following seven polynomial systems holds, 
	\begin{equation}
	\quad \begin{aligned}
	& L=0, L_{345}=0;\\
	& L=0, L_{345}=0, \Gamma_2=\Gamma_3; \\
	& L=0, L_{345}=0, [\Gamma_{1}(\Gamma_{3}+\Gamma_{4})-\Gamma_{2}\Gamma_5](\Gamma_{1}+\Gamma_5)(\Gamma_{2}+\Gamma_{3}+\Gamma_{4})=0;\\
	& L=0, L_{123}=\Gamma_4\Gamma_5; \\
	& L=0, L_{345}=L_{123}; \\
	& L=0, \Gamma_3\Gamma_5=\Gamma_2\Gamma_4; \\
	& L=0, \Gamma_1(\Gamma_3+\Gamma_4)=\Gamma_2\Gamma_5. 
	\end{aligned} \label{equ:constraint_collapse}
	\end{equation}

\end{proof}

To further improve the above result, we first recall some simple tricks to estimate the distances between vorticities when a singular sequence approaches one of the diagrams in Figure \ref{fig:list2}.  

By Proposition \ref{Estimate1}, it is easy to see the following:  A $z w$-edge corresponds to a distance  $r_{kl}\approx \epsilon^2$. The other cases, no edge or a simple edge, i.e., a $z$-edge or a $w$-edge, correspond to a distance $r_{kl}\succ \epsilon^2$.  A   maximal simple edge corresponds to a distance $r_{kl}\approx 1$. A simple edge is maximal if  there is exactly one  of its ends is circled in one of $z$- and $w$-diagram. 
The  non-maximal simple edges  correspond to distances $r_{k l}$ such that $\epsilon^2 \prec r_{k l} \prec 1$.
For distances $r_{k l}$ without edge, it holds $\epsilon^2\prec r_{k l}\preceq \epsilon^{-2}$.  
For a distance $r_{k l}$ without edge that the two ends are not equally circled,  it holds that   $1\prec r_{k l}$, i.e., unbounded. 

We will avoid further case-by-case analysis and use only these simple estimates, along with those stated in Section \ref{sec:diagram&constraints}. We call a 4-product a quantity $p_{i j}=r_{i j}^2 r_{k l}^2 r_{l m}^2 r_{m k}^2$, where $i, j, k, l, m$ are all the indices from 1 to 5 .

\begin{proposition}\label{Prp:estimate_cod-2}
	In  Diagram 5, 6, 7, 9, 10, 13 and 14,   any 4-product is bounded.
\end{proposition}

\begin{proof}
	Note that all distances in Diagram 5, 6, 7, 9, 13 and 14 are bounded. For Diagram 10, we arrive the  estimate by Proposition \ref{Estimate1}. 
\end{proof}

In the following proof, we assume that  $\sum _{j=2}^5\Gamma_j\ne 0$ in Diagram 18, otherwise its constraint is a codimension 2 set.  Similarly, we assume that  it holds in Diagram 19 that  $\Gamma_1+\Gamma_2\ne 0$ and $\Gamma_1+\Gamma_5\ne 0$. Hence, in Diagram 18 and 19,  a simple edge corresponds to a distance $r_{kl}\approx 1$, and  no edge corresponds to a distance $r_{kl}\approx \epsilon^{-2}$.

\begin{proof}[\emph{Proof of Theorem \ref{thm:cod_2}}]
	\emph{Part 1:} For relative equilibria,  repeating the first paragraph of the proof of Theorem \ref{thm:Main} shows that there is a singular  sequence of  complex  
	normalized central configurations approaching 
	one of  the 21 diagrams (the 22 diagrams minus Diagram 17). Since each of Diagram  2, 3, 4, 5, 8, 11, 12, 15, 20, and 21  have two independent constraints, which defines a codimension 2 set, we put them, and all the similar sets obtained by renumbering the five vorticities, in the exceptional set $\mathcal{A}$.  We assume that the singular sequence approaches one of Diagram 1, 6, 7, 9, 10, 13, 14, 16, 18, 19 and 22. 
	
	Suppose that it is Diagram 1. We number the vertices as in figure \ref{fig:list2}. Then $r_{34}^2 \approx \epsilon^4$, and is a dominating function. Let $r_{34}^2\to \infty$. Since all distances in Diagram  6, 7, 9, 13, 14  and 22 are bounded, we are in either Diagram 1 or in Diagram   10, 16, 18, and 19. If we are in Diagram 1, it must be renumbered such that $r_{34}$ is not in the fully edged triangle, so there is a new independent constraint on the five vorticities. In the remaining  cases, there is always a new independent constraint. We add the corresponding codimension 2 sets to the exceptional set $\mathcal {A}$, and  Diagram 1 is now excluded. 
	
	Suppose that it is one of Diagram 6, 7, 9, 10 and 13. We number the vertices as in figure \ref{fig:list2}. Then the 4-product $p_{12} \to 0$, and is a dominating function. Let $p_{12}\to \infty$.   By Proposition \ref{Prp:estimate_cod-2}, we are in one of Diagram  16, 18, and 19. In either cases, there is always  a new independent constraint. We add the corresponding codimension 2 sets to the exceptional set $\mathcal {A}$, and  Diagram 6, 7, 9, 10 and 13  are now excluded. 
	
	Suppose that it is one of Diagram 14 and 22.  Then $r_{34}^2 \approx \epsilon^4$, and is a dominating function. Let $r_{34}^2\to \infty$.   By the estimates from Section \ref{sec:diagram&constraints}, we are in one of Diagram  16, 18, and 19. In either cases, there is also  a new independent constraint. We add the corresponding codimension 2 sets to the exceptional set $\mathcal {A}$, and  Diagram 14 and 22 are now excluded. 
	
	Suppose that it is  one of Diagram 16 and 18.  We number the vertices as in figure \ref{fig:list2}.   Then $r_{14}^2 \approx \epsilon^{-4}$, and is a dominating function. Let $r_{14}^2\to 0$.   We are in Diagram 19. There is   a new independent constraint. We add the corresponding codimension 2 sets to the exceptional set $\mathcal {A}$, and  Diagram 16  and 18 are  now excluded. 
	
	Suppose that it is  Diagram 19.  We number the vertices as in figure \ref{fig:list2}.   Then $r_{34}^2 \to 0$, and is a dominating function. Let $r_{34}^2\to \infty$.   We are in Diagram 19, but it  must be renumbered such that  $r_{34}$ is not a $zw$-edge, so there is a new independent constraint on the five vorticities.  We add the corresponding codimension 2 sets to the exceptional set $\mathcal {A}$. This concludes the construction of $\mathcal{A}$.  The last possibility for a singular sequence is now forbidden. There is no continuum of relative equilibria  if the vorticities  do not belong to $\mathcal{A}$.

	\emph{Part 2:} For collapse configurations,  the  vorticity space is now $\{  (\Gamma_1, \ldots, \Gamma_5): \Gamma_i \in \mathbb{R}^*,  i=1, \ldots, 5,  \sum_{1\le i<j\le 5} \Gamma_i \Gamma_j =0 \}$. 
	Repeating the first paragraph of the proof of Theorem \ref{thm:Main} shows that there is a singular  sequence of
	complex  normalized central configurations approaching 
	one of  the 12 diagrams (Diagram 1, 2, 5-9, 14, 16-19). Since each of Diagram  2 and 8  have two independent constraints, which defines a codimension 2 set, we put them, and all the similar sets obtained by renumbering the five vorticities, in the exceptional set $\mathcal{B}$.   We assume that the singular sequence approaches one of Diagram 1, 5, 6, 7, 9,  14, 16, 17, 18 and 19.

	Suppose that it is Diagram 1. We number the vertices as in figure \ref{fig:list2}. Then $r_{34}^2 \approx \epsilon^4$, and is a dominating function. Let $r_{34}^2\to \infty$. Since all distances in Diagram 5, 6, 7, 9 and 14 are bounded, we are in either Diagram 1 or in Diagram    16, 17, 18, and 19. If we are in Diagram 1, it must be renumbered such that $r_{34}$ is not in the fully edged triangle, so there is a new independent constraint on the five vorticities. In the remaining  cases, there is always a new independent constraint. We add the corresponding codimension 2 sets to the exceptional set $\mathcal {B}$, and  Diagram 1 is now excluded. 
	
	Suppose that it is one of Diagram 5, 6, 7, 9. We number the vertices as in figure \ref{fig:list2}. Then the 4-product $p_{12} \to 0$, and is a dominating function. Let $p_{12}\to \infty$.   By Proposition \ref{Prp:estimate_cod-2}, we are in one of Diagram  16, 17, 18, and 19. In either cases, there is always  a new independent constraint. We add the corresponding codimension 2 sets to the exceptional set $\mathcal {B}$, and  Diagram 5, 6, 7, and 9  are now excluded.

	Suppose that it is  Diagram 19.  We number the vertices as in figure \ref{fig:list2}.   Then $r_{34}^2 \approx \epsilon^{4}$, and is a dominating function. Let $r_{34}^2\to \infty$.   By the estimates from Section \ref{sec:diagram&constraints}, we are in one of Diagram 16, 17, 18, 19. If  we are in Diagram 19,  it  must be renumbered such that  $r_{34}$ is not a $zw$-edge, so there is a new independent constraint on the five vorticities.  In the remaining  cases, there is always a new independent constraint. We add the corresponding codimension 2 sets to the exceptional set $\mathcal {B}$, and  Diagram 19 is now excluded.

	Suppose that it is  Diagram 16. We number the vertices as in figure \ref{fig:list2}. Then $r_{14}^2 \approx \epsilon^{-4}$, and is a dominating function. Let $r_{14}^2\to 0$.  We are in Diagram 14, where $r_{12}^2 r_{13}^2 r_{23}^2 \approx \epsilon^{12}$, and is a dominating function. Let $r_{12}^2 r_{13}^2 r_{23}^2 \to \infty$, then  
	we are in one of Diagram 16, 17, 18. If  we are in Diagram 16,  it  must be renumbered such that  $\{1, 2, 3\}$ is no longer one isolated component of the $w$-diagram, 	so there is a new independent constraint on the five vorticities.  In the remaining  cases, there is always a new independent constraint. We add the corresponding codimension 2 sets to the exceptional set $\mathcal {B}$, and  Diagram 16 is now excluded.

	Suppose that it is  Diagram 17. We number the vertices as in figure \ref{fig:list2}. Then $r_{14}^2 \approx \epsilon^{-4}$, and is a dominating function. Let $r_{14}^2\to 0$.  We are in Diagram 14, where $r_{12}^2 r_{13}^2 r_{23}^2 \approx \epsilon^{12}$, and is a dominating function. Let $r_{12}^2 r_{13}^2 r_{23}^2 \to \infty$, then  
	we are in one of Diagram 17, 18. If  we are in Diagram 17,  it  must be renumbered such that  $\{1, 2, 3\}$ is no longer one isolated component of the $w$-diagram, 	so there is a new independent constraint on the five vorticities.  If it is Diagram 18,  there is  a new independent constraint. We add the corresponding codimension 2 sets to the exceptional set $\mathcal {B}$, and  Diagram 17 is now excluded.

	Suppose that it is  Diagram 18. We number the vertices as in figure \ref{fig:list2}. Then $r_{14}^2 \approx \epsilon^{-4}$, and is a dominating function. Let $r_{14}^2\to 0$.  We are in  Diagram 14, where $r_{23}^2 r_{34}^2 r_{45}^2r_{25}^2 \approx \epsilon^{16}$, and is a dominating function. Let $r_{23}^2 r_{34}^2 r_{45}^2r_{25}^2 \to \infty$, then  
	we are in one copy of  Diagram 18.  It  must be renumbered such that  one of the vertices $\{2, 3, 4, 5\}$ is no longer in the quadrilateral formed by the four simple edges,  
	so there is a new independent constraint on the five vorticities.  We add the corresponding codimension 2 set to the exceptional set $\mathcal {B}$, and  Diagram 18 is now excluded.

	Suppose that it is  Diagram 14.  We number the vertices as in figure \ref{fig:list2}.   Then $r_{14}^2 \approx \epsilon^{4}$, and is a dominating function. Let $r_{14}^2\to \infty$, which is impossible.  This concludes the construction of $\mathcal{B}$.  The last possibility for a singular sequence is now forbidden. There is no continuum of collapse configurations  if the vorticities  do not belong to $\mathcal{B}$.
	
\end{proof}

We also confirm the finiteness under the following restrictions  of Theorem \ref{thm:restri}, namely,
\[  \sum_{j\in J} \Gamma_j\ne 0, \ \sum_{j, k\in J, j\ne k} \Gamma_j\Gamma_k\ne 0, {\rm \ for\  any\ nonempty\  subset\ }  J\ {\rm   of\ } \{1,2,3,4,5\}.   \]

Under the above restriction,  a singular sequence can only approach  Diagram 18 and 19.

\begin{proof}[Proof of  Theorem \ref{thm:restri}]
	Suppose that there are infinitely many solutions  of system \eqref{equ:complexcc} in the complex domain. The argument in the proof of Theorem \ref{thm:Main} implies that there is a singular sequence corresponding to Diagram 18 or 19. In either case, the polynomial
	$\prod_{j\ne k}r_{jk}^{2}$ approach $\infty$, which follows easily
	from the estimation of the distances of Section \ref{sec:diagram&constraints}.  Then it is a
	dominating function.  Push the polynomial to zero and extract a singular sequence. However, the singular sequence would correspond to none of the two diagrams. This is a contradiction.
\end{proof}

\bmhead{Acknowledgements}

We express our sincere gratitude to the reviewers, as their insightful comments and suggestions have significantly contributed to the enhancement of this manuscript. 

\bmhead{Data Availibility Statement} 

All data, models and code generated or used during the study appear in the submitted article.

\section*{Declarations}

On behalf of all authors, the corresponding author states that there is no conflict of interest.

\end{document}